\documentclass[aps,pra,superscriptaddress,nofootinbib,twocolumn]{revtex4-2}

\usepackage[T1]{fontenc}
\usepackage[utf8]{inputenc}
\usepackage{amsmath,amssymb,amsthm,mathtools}
\usepackage{physics}
\usepackage{dsfont}
\usepackage{bm}
\usepackage{microtype}
\usepackage[dvipsnames]{xcolor}
\usepackage[colorlinks=true,citecolor=Mulberry,linkcolor=MidnightBlue,
  urlcolor=MidnightBlue,
  pdftitle={Purification brings advantages in sequential quantum-channel discrimination},
  pdfauthor={Zoe Garcia del Toro, Marco Tulio Quintino, and Jessica Bavaresco}]{hyperref}
\usepackage[nameinlink,capitalize,noabbrev]{cleveref}
\usepackage{amscd}
\usepackage{indentfirst}
\usepackage{graphicx}
\usepackage[american,ngerman,greek,brazilian,british]{babel}
\usepackage{float}

\usepackage{xfrac}
\usepackage{ulem}
\usepackage{enumerate}
\usepackage{sidecap}
\usepackage{enumitem}
\usepackage{bbold}
\usepackage{listings}
\usepackage{algorithm}
\usepackage{algpseudocode}
\usepackage{multirow}
\usepackage{setspace}
\usepackage{lipsum}  
\usepackage{comment}
\usepackage{physics}

\newcommand{\id}{\mathds{1}}

\newcommand{\Hcal}{\mathcal{H}}

\newcommand{\Scal}{\mathcal{S}}

\newcommand{\Ccal}{\mathcal{C}}

\newcommand{\supp}{\operatorname{supp}}

\makeatletter
\newcommand{\appendixtableofcontents}{%
    \section*{Contents}
    \begingroup
    \setcounter{tocdepth}{2} 
    \@starttoc{apc}
    \endgroup
}
\makeatother

\newcommand{\appsection}[1]{%
    \section{#1}%
    \addcontentsline{apc}{section}{\protect\numberline{\thesection}#1}%
}

\newcommand{\appsubsection}[1]{%
    \subsection{#1}%
    \addcontentsline{apc}{subsection}{\protect\numberline{\thesubsection}#1}%
}

\newtheorem{theorem}{Theorem}
\newtheorem*{theorem*}{Theorem}
\newtheorem{proposition}{Proposition}
\newtheorem{lemma}{Lemma}
\newtheorem{corollary}{Corollary}

\begin{document}

\title{Purification enables an unbounded query separation \\
in sequential quantum channel discrimination}

\author{Zoe Garc\'{i}a del Toro}
\email{zoe.garcia@lip6.fr}
\affiliation{Sorbonne Universit\'e, CNRS, LIP6, F-75005 Paris, France}

\author{Marco T\'ulio Quintino}
\email{marco.quintino@lip6.fr}
\affiliation{Sorbonne Universit\'e, CNRS, LIP6, F-75005 Paris, France}

\author{Jessica Bavaresco}
\email{jessica.bavaresco@lip6.fr}
\affiliation{Sorbonne Universit\'e, CNRS, LIP6, F-75005 Paris, France}

\begin{abstract}
Purifications are often used as a convenient mathematical representation of quantum states and channels when constructing protocols for quantum information processing tasks. Although powerful, this perspective can suggest that the purifying degrees of freedom and its correlations with an environment are just a useful change of representation. Here we show how they can instead be exploited as an operational resource in sequential quantum channel discrimination, even when the environmental reference frame is unknown or averaged. In contrast to state discrimination and parallel channel discrimination, for which bare and averaged-purification access are equivalent, sequential access to purified resources exhibits a strict advantage in probability of success. We prove this strict separation via two main results. In the first, we construct a pair of families of qubit-qubit channels with a strict separation in the two-query case, and in the second, we show an example of a pair of qubit-qubit channels whose purifications are perfectly distinguishable with three queries while no finite number of bare (non-purified) queries suffices for perfect discrimination. These separations imply an advantage of sequential over parallel strategies for channel discrimination and establish the impossibility of converting bare channel queries into averaged-purification queries via transformations that allow arbitrary interventions between queries, strengthening previous results.
\end{abstract}

\maketitle

\section{Introduction}
\label{sec:introduction}

Purification is one of the most fundamental constructions in quantum theory. Every mixed quantum state can be regarded as the reduced state of a pure state on a larger Hilbert space, and every quantum channel can be represented, via Stinespring’s dilation theorem~\cite{stinespring1995positive}, by an isometric evolution into a larger Hilbert space followed by tracing out an environment. Such descriptions are often mathematically convenient, as they allow one to reason about quantum information processing protocols entirely in terms of pure states and noiseless, reversible quantum operations. Yet, to recover the original state or channel, the additional degrees of freedom introduced by the purification must ultimately be discarded, even though retaining access to them may reveal information that is absent from the reduced resource. This raises a natural operational question: when does access to a purification provide an advantage over access to the corresponding bare quantum resource?

At first sight, one might expect such an advantage to be immediate. A purification contains the original resource as a marginal and therefore carries at least as much information. However, access to a purification is itself ambiguous: a quantum state or channel does not have a unique purification, since different purifications are related by unitaries acting on the environment system. Furthermore, one may distinguish between access to a fixed but unknown purification of a given state or channel and access to the average over all possible purifications of a given state or channel. While no-go theorems show that the former access model---of a fixed but unknown purification---is inaccessible via transformations acting on bare queries to the reduced quantum state or channel~\cite{Kleinmann2006physical,Kleinmann2007purifying,deltoro2026probabilisticapproximateuniversalquantum,liu2026universalpurificationquantummechanics}, the latter access model---of the average over all possible purifications---has been shown to be achievable via transformations acting on bare queries~\cite{tang2025conjugatequerieshelp,yoshida2026randomdilationsuperchannel,girardi2026randomstinespringsuperchannelconverting}. Although these results demonstrate that these purification access models are genuinely different, it has been shown that for quantum-to-classical tasks under parallel quantum strategies---such as discrimination, property testing, and learning---both the average over all purifications and fixed unknown purifications showcase an equal performance~\cite{chen2025localtestunitarilyinvariant,tang2025conjugatequerieshelp,Chen2025quantumchannel}. Still, their usefulness in quantum information tasks remains an open problem, both in terms of whether purified resources can provide an operational advantage over bare queries in scenarios where sequential strategies are allowed and of whether the two distinct purification access models can themselves be operationally distinguished.

Here we show that the answer is yes. We consider three models of access to a quantum channel: bare access to several independent channel queries; access to the averaged purification queries; and access to queries to a fixed but unknown purification. We first show that, under sequential strategies, access to either purified channel resources can yield a strict advantage over bare queries in channel discrimination tasks. In particular, we show a family of discrimination tasks involving two qubit-qubit channels for which, using only two queries, access to either purification resource brings a strict advantage in the probability of successful discrimination. Next, we show another discrimination task involving two qubit-qubit channels for which perfect discrimination can be achieved with only three queries to either purification resource but cannot be achieved by the bare queries resource with any finite number of copies. This constitutes a unbounded query complexity separation for sequential quantum channel discrimination. 

We present two important consequences of these results: First, they establish a strict advantage of sequential over parallel strategies for the discrimination of nonsignaling channels. Second, they disprove the Conjecture~1.9 of Ref.~\cite{tang2025conjugatequerieshelp}---which expected that bare channel queries could be converted into average purification queries (therein called random Stinespring dilation isometry) via an adaptive, sequential transformation, with some finite query overhead---by implying that no finite number of bare channel queries suffices to prepare three or more queries to the averaged purification resource while allowing arbitrary interventions between queries, strengthening the original result from Ref.~\cite{yoshida2026randomdilationsuperchannel}.

Finally, we show that in quantum-to-quantum tasks, the two purification access models can be operationally distinguished, albeit through a task that was designed specifically to highlight the distinction between the two resources.

\section{Framework}
\label{sec: States, channels, and purification resources }

\subsection{Purification}

We begin by defining the set of possible purifications of a given quantum channel, recovering purification of a quantum state as a particular case.  Quantum channels mapping quantum states on an input Hilbert space $\Hcal_I$ to 
an output Hilbert space $\Hcal_O$ are described by completely positive trace-preserving maps $\mathcal{C}:\mathsf L(\Hcal_I)\rightarrow\mathsf L(\Hcal_O)$ whose Choi operator $C \in \mathsf{L}(\Hcal_{I}\otimes \Hcal_{O})$ is defined as $C\coloneqq \mathcal{I}_I\otimes\mathcal{C}(\ketbra{\Omega})$, where $\mathcal{I}_X$ is the identity map on $\mathsf L(\Hcal_X)$, $\ket{\Omega}\coloneqq\sum_{i=1}^{d_I}\ket{i}\ket{i}$, and $d_X\coloneqq\dim(\Hcal_X)$.

All quantum channels admit an isometric purification: for every quantum channel $\mathcal C$, there exist an environment Hilbert space $\mathcal H_E$ and an isometry $V_C:\mathcal H_I\rightarrow\mathcal H_O\otimes\mathcal H_E$ such that
\begin{align}
\mathcal C(\rho)=\operatorname{tr}_E\!\left[V_C\rho V_C^\dagger\right]
\end{align}
for every quantum state $\rho\in\mathsf{L}(\mathcal{H}_I)$. This fact is known as Stinespring dilation theorem. We denote the corresponding isometric channel by $\mathcal V_C(\rho)=V_C\rho V_C^\dagger$ and refer to it as a \textit{purification} of the quantum channel $\mathcal{C}$. Such condition can be expressed in terms of Choi operators according to
\begin{align}
    \tr_E\ketbra{V_C}&=C,
\end{align}
where $\ket{V_C}\coloneqq(\id_I\otimes V_C)\ket{\Omega}$ is the Choi vector associated to the quantum channel defined by the isometry $V_C$, with $\id_X$ denoting the identity operator in $\Hcal_X$. The set of all possible purifications $\{\ket{V_{C,U}}\}$ of a given Choi operator of a quantum channel $C$, for a fixed environment dimension, can be defined from any fixed purification $\ket{V_C}$ via a unitary $U\in\mathsf{L}(\Hcal_E)$ acting on the environment system, according to
\begin{align}
    \ket{V_{C,U}}&\coloneqq(\id_{IO}\otimes U)\ket{V_C}.
\end{align}
We remark that quantum states can be viewed as the particular case of quantum channels where $d_I=1$; that is, all density matrices are equivalent to Choi operators of quantum channels $C$ with $d_I=1$ and, in particular, all pure quantum states $\ket{\psi}$ are equivalent to Choi vectors of isometric channels $\ket{V_C}$ with $d_I=1$.

\begin{figure}
\begin{center}
    \includegraphics[width=\linewidth]{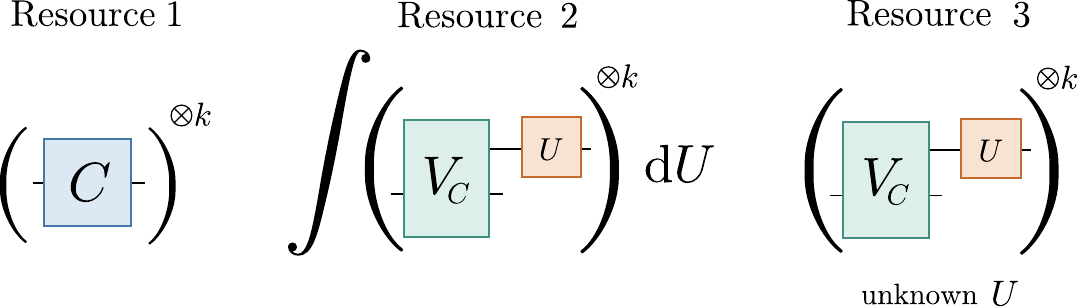}
    \caption{\textbf{Channel access models}. Depiction of the three channel access models. Resource 1 corresponds to $k$ independent queries to channel $C$, defined as ${R}_\mathrm{bare}^{(k)}(C)$ in Eq.~\eqref{eq::R1}. Resource 2 corresponds to the average over the environment of $k$ queries to the same purification of channel $C$, defined as ${R}_\mathrm{avg}^{(k)}(C)$ in Eq.~\eqref{eq::R2}. Resource 3 corresponds to $k$ queries to a fixed unknown purification of channel $C$, defined as ${R}_\mathrm{ukn}^{(k)}(C,U)$ in Eq.~\eqref{eq::R3}.}
    \label{fig:resources}
\end{center}    
\end{figure}

\subsection{Resources and convertibility}
 
In this work, we distinguish three models of access to several queries, or uses, of a quantum channel: via access to the ``bare'' quantum channel through independent queries; access to the average over several queries to the same unknown purification; and access to independent queries to a fixed but unknown purification of a quantum channel (i.e. several queries to the same unknown purification). Formally, these are defined as
\begin{enumerate}[label={R\arabic*.}]
    \item \textit{Bare access.} The available resource is $k$ independent queries to a quantum channel $\mathcal{C}$, mathematically described using Choi operators as
    \begin{align}\label{eq::R1}
        {R}_\mathrm{bare}^{(k)}(C)\coloneqq C^{\otimes k}.
    \end{align}

    \item \textit{Averaged purification access.} The available resource is the Haar average over unitaries acting on the environment of $k$ queries to the same purification of a quantum channel $\mathcal{C}$, mathematically described using Choi operators as
    \begin{align}\label{eq::R2}
        {R}_\mathrm{avg}^{(k)}(C)\coloneqq\int_{\text{Haar}}\!\!\!\!\!\!\!\!\left[(\id_{IO}\otimes U)\ketbra{V_{C}}(\id_{IO}\otimes U^\dagger)\right]^{\otimes k} \mathrm{d}U.
    \end{align}

    \item \textit{Fixed unknown purification access.} The available resource is $k$ independent queries to a fixed but unknown purification of a quantum channel $\mathcal{C}$, mathematically described using Choi operators as
    \begin{align}\label{eq::R3}
        {R}_\mathrm{ukn}^{(k)}(C,U)\coloneqq\left[(\id_{IO}\otimes U)\ketbra{V_{C}}(\id_{IO}\otimes U^\dagger)\right]^{\otimes k},
    \end{align}
    for some unknown unitary $U$ acting on the environment.
\end{enumerate}

These resources are depicted in Fig.~\ref{fig:resources}. Once again, the case of quantum states is regained by setting $d_I=1$ in each of the above access models.

Notice that in resources R2 and R3, the same unitary $U$ acts on every one of the $k$ queries to the purification of $C$, regardless if it is fixed or averaged over; therefore, in both cases the resulting environment systems are correlated.

It is easy to see that these resources satisfy the natural hierarchy 
\begin{align}
    \mathrm{R1}\preceq \mathrm{R2}\prec \mathrm{R3}
\end{align}

in terms of convertibility. That is, there exists a linear, completely positive, and causally well-defined transformation $\Lambda$ (i.e. implementable by a quantum circuit) that converts $\mathrm{R3}\mapsto\mathrm{R2}$ and $\mathrm{R2}\mapsto\mathrm{R1}$. That is, there exists a transformation $\Lambda_{\mathrm{R2}\mapsto\mathrm{R1}}$, satisfying the above conditions, such that
\begin{align}
    \Lambda_{\mathrm{R2}\mapsto\mathrm{R1}}[{R}_\mathrm{avg}^{(k)}(C)] \coloneqq \tr_E [{R}_\mathrm{avg}^{(k)}(C)] = {R}_\mathrm{bare}^{(k)}(C),
\end{align}
for all $k$ and $C$, which is the partial trace. 
Furthermore, there also exists a transformation $\Lambda_{\mathrm{R3}\mapsto\mathrm{R2}}$ satisfying these conditions. This transformation is given by the twirling map, according to
\begin{align}
    &\Lambda_{\mathrm{R3}\mapsto\mathrm{R2}}[R_\mathrm{ukn}^{(k)}(C,U)] \coloneqq \nonumber \\
    &\int_\text{Haar} (\id_{IO}\otimes V^{\otimes k}) R_\mathrm{ukn}^{(k)}(C,U) (\id_{IO}\otimes V^{\otimes k})^\dagger \mathrm{d}V \label{eq: Haar1}\\
    &= \int_\text{Haar} R_\mathrm{ukn}^{(k)}(C,VU) \mathrm{d}V \label{eq: Haar2}\\
    &= \int_\text{Haar} R_\mathrm{ukn}^{(k)}(C,V) \mathrm{d}V \label{eq: Haar3}\\
    &= R_\mathrm{avg}^{(k)}(C)
\end{align}
for all $k$, $C$, and $U$. The equality from Eq.~\eqref{eq: Haar2} and Eq.~\eqref{eq: Haar3} follows from the invariance of Haar measure. In order to recognize that the group average is implementable by a parallel circuit, notice that  applying an arbitrary fixed unitary $V^{\otimes k}$ in the environment space can be done via a parallel quantum circuit, and since the set of parallel quantum circuits is convex, the group average defined in Eq.~\eqref{eq: Haar1} can be implemented by a parallel quantum circuit as well.

Notice that R3 can also be directly converted into R1 via a partial trace in the environment, i.e., 
\begin{align}\label{eq::converts_R3_to_R1}
    \Lambda_{\mathrm{R3}\mapsto\mathrm{R1}}[{R}_\mathrm{ukn}^{(k)}(C,U)] \coloneqq \tr_E [{R}_\mathrm{ukn}^{(k)}(C,U)] = {R}_\mathrm{bare}^{(k)}(C),
\end{align}
for all $k$, $C$, and $U$.

At the same time, conversion is not possible from $\mathrm{R1}$ to $\mathrm{R3}$ and from $\mathrm{R2}$ to $\mathrm{R3}$: following arguments presented for quantum state purification in Refs.~\cite{Kleinmann2006physical,Kleinmann2007purifying,liu2026universalpurificationquantummechanics}, and similarly for quantum channel purification in Refs.~\cite{deltoro2026probabilisticapproximateuniversalquantum}, the only linear positive map that can decrease the rank of an operator is the one that discards its input and prepares a constant output. Hence, no linear positive map can universally ``purify'' resources R1 and R2 into R3 for every quantum channel $\mathcal{C}$. In other words, since $\rank[{R}_\mathrm{ukn}^{(k)}(C,U)]=1$ for all $k$, $C$, and $U$, and in general $\rank[{R}_\mathrm{avg}^{(k)}(C)]>1$ and $\rank[{R}_\mathrm{bare}^{(k)}(C)]>1$, it follows from the same arguments presented in \cite{Kleinmann2006physical,Kleinmann2006physical,deltoro2026probabilisticapproximateuniversalquantum} that 
\begin{align}
\begin{split}
    \nexists \ &\text{linear and positive} \ \Lambda_{\mathrm{R1}\mapsto\mathrm{R3}}; \\
    &\Lambda_{\mathrm{R1}\mapsto\mathrm{R3}}[{R}_\mathrm{bare}^{(k)}(C)] =  {R}_\mathrm{ukn}^{(k)}(C,U)
\end{split}    
\end{align}
for all $k$, $C$, and some $U$, and similarly,
\begin{align}
\begin{split}\label{eq::nonconversionR2_to_R3}
    \nexists \ &\text{linear and positive} \ \Lambda_{\mathrm{R2}\mapsto\mathrm{R3}}; \\
    &\Lambda_{\mathrm{R2}\mapsto\mathrm{R3}}[{R}_\mathrm{avg}^{(k)}(C)] =  {R}_\mathrm{ukn}^{(k)}(C,U)
\end{split}    
\end{align}
for all $k$, $C$, and some $U$.
The convertibility of $\mathrm{R1}\mapsto\mathrm{R2}$ depends on the causal constraints imposed on the quantum circuit that implements such transformation, and in particular, on whether it allows for arbitrary interventions in between uses~\cite{yoshida2026randomdilationsuperchannel,girardi2026randomstinespringsuperchannelconverting}. This point will be discussed in more detail in Sec.~\ref{sec:discussion}.

We now show how this hierarchy of convertibility between different access models can be exploited as a resource for general estimation/learning/property testing tasks, which we here refer to as quantum-to-classical tasks. In particular, we show that access to purification resources R2 and R3 brings strict advantages to sequential quantum channel discrimination.

\subsection{Quantum-to-classical tasks}
\label{sec: discrimination tasks}

Let $(\mathsf{X},\Sigma_{\mathsf X})$ be a measurable space, referred to as the hypothesis set, let $\{p(x)\}_{x\in\mathsf{X}}$ be the prior distribution over the hypothesis set, and let $(\mathsf{Y},\Sigma_{\mathsf Y})$ be a measurable space, referred to as the decision set. A bounded measurable function\footnote{Here, boundedness means that there exists a finite constant $M\in\mathbb{R}$ such that $|s(x,y)|\leq M$ for all $(x,y)\in\mathsf{X}\times\mathsf{Y}$. Together with measurability, this ensures we have no convergence problems, that is, average-score integral is absolutely convergent for any prior probability distribution and any measurable conditional decision distribution. Note that when $\mathsf{X}$ and $\mathsf{Y}$ are finite test, real-valued score functions are automatically bounded. }
\begin{align}
    s:\mathsf{X}\times\mathsf{Y}\longrightarrow\mathbb R
\end{align}
assigns a score value $s(x,y)$ when the true hypothesis is $x$ and the
decision is $y$. This abstract framework captures a general family of ``quantum-to-classical'' tasks, such as parameter estimation, learning, and hypothesis and property testing. These are tasks where the relevant figure of merit depends on some classical information encoded into quantum systems, and which is decoded via a quantum strategy, that culminates into a quantum measurement yielding a classical value. 
In general, an interesting figure of merit in such tasks is the average score given by the expression 
\begin{align}
    \mathrm{Score}(p,s,\mathrm{Pr}) \coloneqq
    \int_\mathsf{X} \int_\mathsf{Y} p(x) s(x,y) \mathrm{Pr}(y|x) \mathrm{d}x\mathrm{d}y,
\end{align}
where the decision function $\mathrm{Pr}(y|x)$ is given by the probability of outputting decision $y$ conditioned on a given hypothesis $x$, which is a function of the quantum strategy typically expressed as generalized Born rule.

Take as an example the particular case of quantum state discrimination. In this case, the number of hypotheses $x$ is countable and finite, and the set of decisions can be taken to be the same as the set of hypothesis $\mathsf{Y}=\mathsf{X}$ without loss of generality. The score function is the all-or-nothing function 
\begin{align}
    s(x,y)=\delta_{x,y} 
\end{align}
given by the Kronecker delta function, which assigns score of $1$ to correct decisions and of $0$ to incorrect ones. Each hypothesis $x\in\mathsf{X}$ corresponds to the label of a quantum state $\{\rho_x\}_{x\in\mathsf{X}}$ and the quantum strategy, given by a positive operator-valued measure (POVM) $\{M_y\}_{y\in\mathsf{Y}}$, defines the function
\begin{align}
    \mathrm{Pr}(y|x)=\tr(\rho_xM_y).  
\end{align}
The figure of merit is then the maximum probability of successful discrimination of an ensemble $\{p(x),\rho_x\}$, given by
\begin{align}
    P_\mathrm{succ}(\{p(x),\rho_x\}) &\coloneqq\sup_{\{M_y\}}\mathrm{Score}(p,s,\mathrm{Pr}) \\
    &= \sup_{\{M_y\}} \sum_{x,y} p(x)\delta_{x,y} \tr(\rho_xM_y) \\
    &= \sup_{\{M_x\}} \sum_{x} p(x)\tr(\rho_xM_x).
\end{align}

Another example of a quantum-to-classical task is the task of quantum state estimation, in which a common choice of score is the fidelity between the unknown state $\rho_x$ and the state estimate $\sigma_y$ associated with the measurement outcome $y\in\mathsf{Y}$~\cite{Hayashi1998,Bagan2006Estimation}. Taking the decision set $\mathsf{Y}$ to label the state estimates and the score function to be the fidelity $s(x,y)=F(\rho_x,\sigma_y)$, the optimal average fidelity is given by
\begin{align}
    \langle F(\{p(x),\rho_x\})\rangle &\coloneqq \sup_{\{M_y\}}\mathrm{Score}(p,s,\mathrm{Pr}) \\
    &= \sup_{\{M_y\}} \int_\mathsf{X}\int_\mathsf{Y} p(x) F(\rho_x,\sigma_y) \tr(\rho_xM_y) \mathrm{d}x\mathrm{d}y.
\end{align}

Here, we consider quantum-to-classical tasks where the hypotheses are encoded into quantum channels that are available in one of the three access models we defined. That is, the hypotheses and their associated priors are given in the form of ensembles 
\begin{align}
    \mathsf{C}_\mathrm{bare}^{(k)} &\coloneqq \{p(x), {R}_\mathrm{bare}^{(k)}(C_x)\}_{x\in\mathsf{X}} \label{eq::ensemble_bare} \\
    \mathsf{C}_\mathrm{avg}^{(k)}  &\coloneqq \{p(x), {R}_\mathrm{avg}^{(k)}(C_x)\}_{x\in\mathsf{X}} \label{eq::ensemble_avg} \\
    \mathsf{C}_\mathrm{ukn}^{(k)}(\{U_x\})  &\coloneqq \{p(x), {R}_\mathrm{ukn}^{(k)}(C_x,U_x)\}_{x\in\mathsf{X}}. \label{eq::ensemble_unk}
\end{align}
where resources ${R}_\mathrm{bare}^{(k)}(C_x)$, ${R}_\mathrm{avg}^{(k)}(C_x)$, and ${R}_\mathrm{ukn}^{(k)}(C_x,U_x)$ are defined in Eqs.~\eqref{eq::R1}--\eqref{eq::R3}.

\begin{figure}
\begin{center}
     \includegraphics[width=0.9\linewidth]{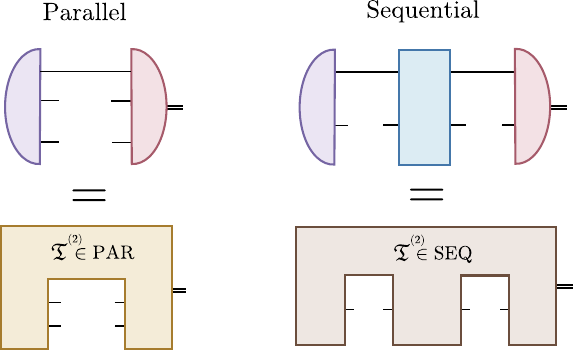}
     \caption{\textbf{Implementation of two-query quantum testers with quantum circuits.} Top: a quantum strategy for a quantum-to-classical task that acts on two queries to a channel resource either in parallel or sequentially, implemented by a quantum circuit. Bottom: the corresponding representation of these quantum circuits as parallel and sequential quantum testers $\mathfrak{T}^{(2)}\in\mathrm{PAR}$ and $\mathfrak{T}^{(2)}\in\mathrm{SEQ}$.}
     \label{fig:testers}
 \end{center}
 \end{figure}

\subsection{Parallel and sequential quantum strategies}

We consider two kinds of quantum strategies that can probe the quantum channels available across the different access models: Non-adaptive ones, called \textit{parallel strategies}, which can be implemented by quantum circuits that simultaneously query the available resources, and adaptive ones, called \textit{sequential strategies}, which can be implemented by quantum circuits that query the available resources sequentially, allowing for arbitrary quantum interventions between queries. Formally, these are characterized by parallel and sequential quantum testers~\cite{Chiribella2008Memory,Ziman2008process,bavaresco2021strict,bavaresco2022unitary}, which are depicted for the case of two queries in Fig.~\ref{fig:testers}.

Parallel quantum testers are described by a set of linear operators $\mathfrak T^{(k)} = \{T_y^{(k)}\}_{y\in\mathsf{Y}} \in \mathrm{PAR}$, where $T_y^{(k)}\in\mathsf{L}(\bigotimes_{i=1}^k\Hcal_{I_i}\otimes\Hcal_{\widetilde{O}_i})$, with $\Hcal_{\widetilde{O}_i}\coloneqq\Hcal_{O_i}\otimes\Hcal_{E_i}$ (note that in some cases we will take $d_{E_i}=1$), are positive semidefinite operators, and $T^{(k)}\coloneqq\int_\mathsf{Y}T_y^{(k)}\mathrm{d}y$ satisfies
\begin{align} \label{eq:PAR_tester}
    \tr(T^{(k)}) = \prod_{i=1}^k d_{\widetilde{O}_i}, \quad T^{(k)} = {}_{\widetilde{O}_1\ldots \widetilde{O}_k}T^{(k)},
\end{align}
where ${}_{X}W\coloneqq\tr_X(W)
\otimes\frac{\id_X}{d_X}$ is the trace-and-replace map, which traces the system acting on $\Hcal_X$, and replaces it by the normalized identity $\frac{\id_X}{d_X}$.  Such conditions guarantee that every parallel quantum tester can be realized by preparing a common, possibly entangled, input state across all channel queries, followed by a final joint quantum measurement on the resulting outputs~\cite{Chiribella2008Memory,Ziman2008process}.

Sequential quantum testers are described by a set of linear operators $\mathfrak T^{(k)} = \{T_y^{(k)}\}_{y\in\mathsf{Y}} \in \mathrm{SEQ}$, where $T_y^{(k)}\in\mathsf{L}(\bigotimes_{i=1}^k\Hcal_{I_i}\otimes\Hcal_{\widetilde{O}_i})$ are positive semidefinite operators and $T^{(k)}\coloneqq\int_\mathsf{Y}T_y^{(k)}\mathrm{d}y$ satisfies
\begin{align} \label{eq:SEQ_tester}
\begin{split}
    &\tr(T^{(k)}) = \prod_{i=1}^k d_{\widetilde{O}_i}, \quad T^{(k)} = {}_{\widetilde{O}_k}T^{(k)}, \\
    &{}_{I_j\widetilde{O}_j\cdots I_k\widetilde{O}_k}T^{(k)}={}_{\widetilde{O}_{j-1}I_j\widetilde{O}_j\cdots I_k\widetilde{O}_k}T^{(k)}, \ \ \forall \ j\in\{2,\ldots,k\}.
\end{split}
\end{align}
Such conditions guarantee that every sequential quantum tester can be realized by an initial state preparation, followed by a sequence of intermediate quantum channels that may carry quantum memory between consecutive queries, forming what is called a quantum comb~\cite{Chiribella2009Networks,Chiribella2008quantum}, and a final quantum measurement~\cite{Chiribella2008Memory,Ziman2008process}.

Let $\mathrm{QS}\in\{\mathrm{PAR},\mathrm{SEQ}\}$ denote the kind of quantum strategy. For a general quantum-to-classical task defined by an arbitrary score function $s$, we define the maximum average score attained by a quantum strategy $\mathrm{QS}$ in each access model. For the access model related to resource R1, we have that
\begin{align}
    &P_\mathrm{QS}(\mathsf{C}_\mathrm{bare}^{(k)};s) \nonumber \\
    &\coloneqq \sup_{\{T_y^{(k)}\}\in\mathrm{QS}} \int_\mathsf{X} \int_\mathsf{Y} p(x) s(x,y) \tr[{R}_\mathrm{bare}^{(k)}(C_x)\,T_y^{(k)}] \mathrm{d}x\mathrm{d}y \\
    &= \sup_{\{T_y^{(k)}\}\in\mathrm{QS}} \int_\mathsf{X} \int_\mathsf{Y} p(x) s(x,y) \tr[C_x^{\otimes k}\,T_y^{(k)}] \mathrm{d}x\mathrm{d}y.
\end{align}
For the access model related to resource R2,
\begin{align}
    &P_\mathrm{QS}(\mathsf{C}_\mathrm{avg}^{(k)};s) \nonumber \\
    &\coloneqq \sup_{\{T_y^{(k)}\}\in\mathrm{QS}} \int_\mathsf{X} \int_\mathsf{Y} p(x) s(x,y) \tr[{R}_\mathrm{avg}^{(k)}(C_x)\,T_y^{(k)}] \mathrm{d}x\mathrm{d}y \\
    &= \sup_{\{T_y^{(k)}\}\in\mathrm{QS}} \int_\mathsf{X} \int_\mathsf{Y}\mathrm{d}x\mathrm{d}y \,p(x) s(x,y) \nonumber\\
    &\times \tr\left[\int_\mathrm{Haar} \!\!\!\!\!\!\!\! \mathrm{d}U_x\left[(\id_{IO}\otimes U_x)\ketbra{V_{C_x}}(\id_{IO}\otimes U_x^\dagger)\right]^{\otimes k}\,T_y^{(k)}\right] .
\end{align}
Finally, for the access model related to resource R3, we consider the \textit{worst-case scenario} across all possible unknown unitaries acting on the environment, yielding
\begin{align}
    &P_\mathrm{QS}(\mathsf{C}_\mathrm{ukn}^{(k)};s) \nonumber \\
    &\coloneqq \sup_{\{T_y^{(k)}\}\in\mathrm{QS}} \min_{\{U_x\}}\int_\mathsf{X} \int_\mathsf{Y} p(x) s(x,y) \nonumber\\
    &\quad\quad\quad\quad\quad\quad\quad\quad\times \tr[{R}_\mathrm{ukn}^{(k)}(C_x,U_x)\,T_y^{(k)}] \mathrm{d}x\mathrm{d}y \\
    &=\sup_{\{T_y^{(k)}\}\in\mathrm{QS}} \min_{\{U_x\}} \int_\mathsf{X} \int_\mathsf{Y}\mathrm{d}x\mathrm{d}y \,p(x) s(x,y) \nonumber\\
    &\times \tr[\left[(\id_{IO}\otimes U_x)\ketbra{V_{C_x}}(\id_{IO}\otimes U_x^\dagger)\right]^{\otimes k}\,T_y^{(k)}] .
\end{align}

The equivalence between Haar-averaged and worst-case unknown purification access in quantum-to-classical tasks follows from a standard Haar-symmetrization argument, with roots in Holevo's approach to covariant quantum estimation~\cite{Holevo1982}, and its tester analogue~\cite{Chiribella2008Covariant}. The lemma below extends the group average argument from Lemma~3.8  of Ref.~\cite{Chen2025quantumchannel} to sequential strategies. For completeness, we present a short proof in our notation and framework, covering both parallel and sequential strategies.

\begin{lemma}\label{lemma::avg_equal_unk_sequential}
    For every score function $s$, every finite number of queries $k$, all quantum channel ensembles $\mathsf{C}_\mathrm{avg}^{(k)}$ and $\mathsf{C}_\mathrm{ukn}^{(k)}(\{U_x\})$, and strategy classes $\mathrm{QS} \in \{\mathrm{PAR},\mathrm{SEQ}\}$, it holds that
    \begin{align}
        P_{\mathrm{QS}}(\mathsf{C}_\mathrm{avg}^{(k)};s)
        =P_{\mathrm{QS}}(\mathsf{C}_\mathrm{ukn}^{(k)};s).
    \end{align}
\end{lemma}

\begin{proof}
   We start by noticing that, since for every strategy the average score is at least as large as the worst-case score, it holds that 
   $ P_{\mathrm{QS}}(\mathsf{C}_\mathrm{ukn}^{(k)};s) \leq P_{\mathrm{QS}}(\mathsf{C}_\mathrm{avg}^{(k)};s)$.

   We will now show that, for parallel or sequential strategies, it holds that $P_{\mathrm{QS}}(\mathsf{C}_\mathrm{avg}^{(k)};s)
    \leq P_{\mathrm{QS}}(\mathsf{C}_\mathrm{ukn}^{(k)};s)$. Let $T_y^{*(k)}$ be the parallel or sequential tester that attains $P_{\mathrm{QS}}(\mathsf{C}_\mathrm{avg}^{(k)};s)$. We define the new tester
    \begin{align}
        {T'}_y^{(k)}\coloneqq\int_\text{Haar} (\id_{IO}\otimes V^{\otimes k}) T_y^{*(k)}(\id_{IO}\otimes V^{\otimes k})^\dagger \mathrm{d}V.
    \end{align}
    Due to the invariance of the Haar measure [see Eq.~\eqref{eq: Haar2}], it holds that
    \begin{align}
        &  \int_\mathsf{X} \int_\mathsf{Y} p(x) s(x,y) \tr[{R}_\mathrm{avg}^{(k)}(C_x)\, {T}_y^{*(k)}] \mathrm{d}x\mathrm{d}y \\
        =& \int_\mathsf{X} \int_\mathsf{Y} p(x) s(x,y) \tr[{R}_\mathrm{avg}^{(k)}(C_x)\, {T'}_y^{(k)}] \mathrm{d}x\mathrm{d}y  \\ 
        =&  \int_\mathsf{X} \int_\mathsf{Y} p(x) s(x,y) \tr[{R}_\mathrm{ukn}^{(k)}(C_x)\, {T'}_y^{(k)}] \mathrm{d}x\mathrm{d}y . 
    \end{align}
    That is, the new tester $\{{T'}_y^{(k)}\}_y$ attains the same score $P_\mathrm{QS}(\mathsf{C}_\mathrm{avg}^{(k)};s)$ of $\{{T}_y^{*(k)}\}_y$ in the average case, and the same score of $\{{T}_y^{*(k)}\}_y$ in the worst case, the one corresponding to 
    $P_\mathrm{QS}(\mathsf{C}_\mathrm{ukn}^{(k)};s)$.

    We finish the proof by showing that, if $\{T_y^{*(k)}\}\in\mathrm{QS}$ then $\{{T'}_y^{(k)}\}\in\mathrm{QS}$. Since $T_y^{*(k)}\geq0$, it follows that ${T'}_y^{(k)}\geq0 $.
    Also, direct calculation shows that, if $\{{T}_y^{*(k)}\}\in\mathrm{PAR}$, then  $\{{T'}_y^{(k)}\}$ respects the parallel tester constraints of Eq.~\eqref{eq:PAR_tester}. Analogously, if $\{{T}_y^{*(k)}\}\in\mathrm{SEQ}$, then  $\{{T'}_y^{(k)}\}$ respects the sequential tester constraints of Eq.~\eqref{eq:SEQ_tester}. 
\end{proof}

We remark that the proof of Lemma~\ref{lemma::avg_equal_unk_sequential} holds also in the case of general testers~\cite{bavaresco2021strict}, which correspond to well-defined transformations that do not necessarily respect a definite causal order~\cite{costa2026indefinitequantumcausality} (see e.g. App. B of Ref.~\cite{bavaresco2022unitary} for an explicit calculation).

\section{No advantage of purification resources in quantum-to-classical tasks with parallel strategies}

It follows from Theorem 3.1 of Ref.~\cite{Chen2025quantumchannel} that, for parallel strategies, access to queries to a fixed unknown purification offers no advantage over access to bare queries in quantum-to-classical tasks. Here, we revisit this result using our notation and approach, presenting a short, direct proof.

\begin{proposition}[No purification advantage for parallel strategies]
\label{prop:unk-eq-avg-channel}
    For every score function $s$, every finite number of queries $k$, and all quantum channel ensembles $\mathsf{C}_\mathrm{bare}^{(k)}$, $\mathsf{C}_\mathrm{avg}^{(k)}$, and $\mathsf{C}_\mathrm{ukn}^{(k)}(\{U_x\})$, parallel quantum strategies yield
    \begin{align}\label{eq::thm1}
        P_{\mathrm{PAR}}(\mathsf{C}_\mathrm{bare}^{(k)};s)
        =P_{\mathrm{PAR}}(\mathsf{C}_\mathrm{avg}^{(k)};s)
        =P_{\mathrm{PAR}}(\mathsf{C}_\mathrm{ukn}^{(k)};s).
    \end{align}
\end{proposition}

\begin{proof}
    The equality $P_{\mathrm{PAR}}(\mathsf{C}_\mathrm{avg}^{(k)};s)
    =P_{\mathrm{PAR}}(\mathsf{C}_\mathrm{ukn}^{(k)};s)$ follows from Lemma~\ref{lemma::avg_equal_unk_sequential}. To prove the equality $P_{\mathrm{PAR}}(\mathsf{C}_\mathrm{bare}^{(k)};s)
    =P_{\mathrm{PAR}}(\mathsf{C}_\mathrm{avg}^{(k)};s)$, first notice that, since the averaged purification queries can be converted into the bare channel queries via a partial trace, using a data processing inequality it follows that 
    \begin{align}
        P_{\mathrm{PAR}}(\mathsf{C}_\mathrm{bare}^{(k)};s)
        \leq P_{\mathrm{PAR}}(\mathsf{C}_\mathrm{avg}^{(k)};s).
    \end{align}
    Now, we invoke results from Refs.~\cite{yoshida2026randomdilationsuperchannel,girardi2026randomstinespringsuperchannelconverting}, which adapted the ``acorn trick'' method for quantum states of  Ref.~\cite{tang2025conjugatequerieshelp} and showed the following: There exists a transformation $\Lambda_{\mathrm{R1}\mapsto\mathrm{R2}}$ that can be implemented by a parallel quantum circuit and that simultaneously satisfies the two following properties:
   \begin{enumerate}[label={(\alph*)}]
        \item $\Lambda_{\mathrm{R1}\mapsto\mathrm{R2}}\left[{R}_\mathrm{bare}^{(k)}(C)\right] = {R}_\mathrm{avg}^{(k)}(C) \ \ \forall \ k,C$
        \item $\forall \ \{T^{(k)}_y\}_y\in\mathrm{PAR}, \quad \left\{\Lambda_{\mathrm{R1}\mapsto\mathrm{R2}}^\dagger\left(T^{(k)}_y\right)\right\}_y\in\mathrm{PAR}$,
    \end{enumerate}
    where $\Lambda_{\mathrm{R1}\mapsto\mathrm{R2}}^\dagger$ is the adjoint map of $\Lambda_{\mathrm{R1}\mapsto\mathrm{R2}}$. Point (b) is a consequence of the fact that $\Lambda_{\mathrm{R1}\mapsto\mathrm{R2}}$ 
    can be implemented by a quantum circuit that queries the resources in parallel, using a single encoder and single decoder channel, which has the same structure of parallel testers. Hence, $\tr[{R}_\mathrm{avg}^{(k)}(C_x)\,T_y^{(k)}] = \tr[\Lambda_{\mathrm{R1}\mapsto\mathrm{R2}}\left({R}_\mathrm{bare}^{(k)}(C_x)\right)\,T_y^{(k)}] = \tr[{R}_\mathrm{bare}^{(k)}(C_x)\,\Lambda_{\mathrm{R1}\mapsto\mathrm{R2}}^\dagger\left(T_y^{(k)}\right)] = \tr[{R}_\mathrm{bare}^{(k)}(C_x)\,{T'}_y^{(k)}]$ where $\{{T'}_y^{(k)}\}_y\in\mathrm{PAR}$ attains the same score as $\{T_y^{(k)}\}_y\in\mathrm{PAR}$. Consequently, via a data processing inequality, it follows that 
    \begin{align}
        P_{\mathrm{PAR}}(\mathsf{C}_\mathrm{bare}^{(k)};s)
        \geq P_{\mathrm{PAR}}(\mathsf{C}_\mathrm{avg}^{(k)};s),
    \end{align}
    which allows us to conclude Eq.~\eqref{eq::thm1}.
\end{proof}

An analogue result of Proposition~\ref{prop:unk-eq-avg-channel} for quantum states, showing the lack of an advantage from purified samples was previously established in the context of quantum property testing~\cite{chen2025localtestunitarilyinvariant}. Later, this result was re-derived in Ref.~\cite{tang2025conjugatequerieshelp}, and then simplified in Ref.~\cite{girardi2026randompurificationchannelsimple}. Since quantum states can be formally identified as quantum channels with $d_I=1$, the lack of advantage for purified state resources for quantum-to-classical tasks can be recognized as a particular case of Proposition~\ref{prop:unk-eq-avg-channel}:

\begin{corollary}
    For any quantum-to-classical task in which the hypotheses are encoded in quantum states, the three access models—$(\mathrm{R1})$ access to bare copies of the state, $(\mathrm{R2})$ access to the average over all purifications, and $(\mathrm{R3})$ access to a fixed unknown purification—yield the same performance for any score function $s$, any finite number of queries $k$, and any ensemble of quantum states.
\end{corollary}

\section{Advantages of purification resources in quantum channel discrimination with sequential strategies}
\label{sec: Discrimination advantage}

We now show that in the case of sequential strategies, this equivalence between resources breaks down. In particular, while the averaged purification and fixed unknown purifications remain equivalent resources in all quantum-to-classical tasks with sequential strategies (Lemma~\ref{lemma::avg_equal_unk_sequential}), we establish two complementary separations between bare access and purified access, for tasks of quantum channel discrimination with sequential strategies. The first separation result (Theorem~\ref{thm:explicit-2copies-gap}) involves the sequential discrimination of a pair of entanglement-breaking qubit-qubit channels, with a uniform prior, for which the optimal probability of successful discrimination for $k=2$ queries is strictly greater for access models ${R}_\mathrm{ukn}^{(k)}(C,U)$ and ${R}_\mathrm{avg}^{(k)}(C)$ than for access model ${R}_\mathrm{bare}^{(k)}(C)$. In the second separation result (Theorem~\ref{thm: perfect discrimination 3 copies}), we show a pair of qubit-qubit channels that can be perfectly sequentially discriminated with $k=3$ queries in the either of the two purified access models ${R}_\mathrm{ukn}^{(k)}(C,U)$ and ${R}_\mathrm{avg}^{(k)}(C)$, but that, nonetheless, cannot be perfectly discriminated in the bare access model ${R}_\mathrm{bare}^{(k)}(C)$ for any finite number of queries. Thus, unlike state discrimination and parallel channel discrimination (which are particular cases of quantum-to-classical tasks under parallel strategies), purification constitutes a genuine resource in the sequential setting, even when averaging over the environment freedom or for an unknown environment freedom.

\begin{theorem}[Two-query sequential channel discrimination advantage from purification]\label{thm:explicit-2copies-gap}
    In a channel discrimination task, where the score function is $s(x,y)=\delta_{x,y}$, for $k=2$ queries, there exist channel ensembles such that, for sequential strategies,
    \begin{align}
        P_{\mathrm{SEQ}}(\mathsf{C}_\mathrm{bare}^{(2)};s)
        < P_{\mathrm{SEQ}}(\mathsf{C}_\mathrm{avg}^{(2)};s) = P_{\mathrm{SEQ}}(\mathsf{C}_\mathrm{ukn}^{(2)};s).
    \end{align}
    \\
    
    In particular, let $\Hcal_I\simeq\Hcal_O\simeq\mathbb C^2$ and $\Hcal_E\simeq\mathbb C^{d_E}$, with $d_E\geq4$, and let $X,Y,Z$ be the Pauli matrices. On $\Hcal_I\otimes\Hcal_O$, set
    \begin{align}
        H\coloneqq X_I\otimes Z_O-Y_I\otimes Y_O-\id_I\otimes X_O.
    \end{align}
    Now take the hypothesis set $\mathsf{X}=\{+,-\}$ and the decision set $\mathsf{Y}=\mathsf{X}$. Encode the hypothesis set into the family of qubit-qubit channels with Choi operators defined as
    \begin{align}
        C_\pm(t)\coloneqq\frac{\id_{IO}}{2}\pm tH
        \in\mathsf L(\Hcal_I\otimes\Hcal_O),
    \label{eq:condition-positivity-Cpm}
    \end{align}
    for $0<t<\frac{1}{2\sqrt5}$.
    Then, set the channel ensembles $\mathsf{C}_\mathrm{bare}^{(2)}(t) \coloneqq \{p(\pm),R^{(2)}_\mathrm{bare}(C_\pm(t))\}$ and $\mathsf{C}_\mathrm{avg}^{(2)}(t) \coloneqq \{p(\pm),R^{(2)}_\mathrm{avg}(C_\pm(t))\}$   
    according to the two access models in Eqs.~\eqref{eq::R1} and~\eqref{eq::R2} and by taking an equal prior distribution $p(+)=p(-)=1/2$. Then, considering sequential strategies,
    \begin{align}
        P_{\mathrm{SEQ}}(\mathsf{C}_\mathrm{avg}^{(2)}(t);s)
        -P_{\mathrm{SEQ}}(\mathsf{C}_\mathrm{bare}^{(2)}(t);s)
        \geq \frac{7-3\sqrt{5}}{6}\,t>0.
    \end{align}
\end{theorem}

The complete proof is given in Appendix~\ref{ap:: qubit advantage}. There, we also show how this bound for the gap in probability of success, expressed as a function of the parameter $t$, can be numerically verified to be tight. Here we discuss the reasoning of the proof.

The gap arises because averaging over the unknown purification frame does not remove all environment-level information. Environmental correlations are organized into symmetric and antisymmetric sectors. A sequential tester can preserve information from the first purified output in a quantum memory and combine it with the second output in a sector-dependent way. On the other hand, bare-channel access discards the environment systems and therefore loses precisely this information. Under sequential strategy constraints, the two bare queries cannot always compensate for this loss.

The proof makes this mechanism rigorous by comparing the semidefinite programs associated with the bare and average access models. For the averaged purification resource, the symmetric-antisymmetric decomposition yields a feasible sequential tester whose two environment-sector contributions combine constructively. This gives an achievable lower bound on the corresponding discrimination probability. For the bare-queries resource, we construct a feasible dual certificate satisfying the sequential causal constraints. This certificate gives an upper bound valid for every sequential strategy using two bare channel queries. 
    
We find that the lower bound obtained with the averaged purification resource is strictly larger than the upper bound obtained with the bare channel queries. Weak duality therefore proves the discrimination gap without requiring either certificate to be optimal.
\\

Next, we strengthen the separation by exhibiting a task for which access to the averaged purification resource or a fixed unknown purification resource enables perfect discrimination with $k=3$ queries, at the same time showing that this is unachievable for any finite number of bare channel queries.

\begin{theorem}[Unbounded query complexity separation from purification]
\label{thm: perfect discrimination 3 copies}
    In a channel discrimination task, where the score function is $s(x,y)=\delta_{x,y}$, there exist channel ensembles such that, for sequential strategies,
    \begin{align}\label{eq::inf_separation}
        P_{\mathrm{SEQ}}(\mathsf{C}_\mathrm{bare}^{(k)};s)
        < P_{\mathrm{SEQ}}(\mathsf{C}_\mathrm{avg}^{(3)};s) 
        = P_{\mathrm{SEQ}}(\mathsf{C}_\mathrm{ukn}^{(3)};s) = 1.
    \end{align}
    for all finite $k$.
    \\

    In particular, let $\Hcal_I\simeq\Hcal_O\simeq\mathbb C^2$ and $\Hcal_E\simeq\mathbb C^{d_E}$, with $d_E\geq3$. Now take the hypothesis set $\mathsf{X}=\{0,1\}$ and the decision set $\mathsf{Y}=\mathsf{X}$. Encode the hypothesis set into the pair of qubit-qubit channels with Choi operators given by
    \begin{align}
        C_0&\coloneqq\ketbra{00}+\ketbra{11}, \label{eq::C0} \\
        C_1&\coloneqq\frac12\ketbra{-0}+\ketbra{+1}+\frac12\ketbra{-1}, \label{eq::C1} 
    \end{align}
    where $\ket{\pm}=(\ket0\pm\ket1)/\sqrt2$. 
    Then, set the channel ensembles $\mathsf{C}_\mathrm{bare}^{(k)}$, $\mathsf{C}_\mathrm{avg}^{(3)}$, and $\mathsf{C}_\mathrm{ukn}^{(3)}(\{U_x\})$ 
    according to the two access models in Eqs.~\eqref{eq::R1}--\eqref{eq::R3} and by taking channels $C_0$ and $C_1$ and an equal prior distribution $p(0)=p(1)=1/2$. These ensembles satisfy Eq.~\eqref{eq::inf_separation}, i.e., can be perfectly discriminated with 3 queries to purified channel resources but cannot be perfectly discriminated with any finite number of queries to a bare channel resources.
\end{theorem}

The complete proof is given in Appendix~\ref{app:: perfect discrimination k copies}. The first part of the proof relies on the following lemma:

\begin{lemma}
    The quantum channels with Choi operators $C_0$ and $C_1$ in Eqs.~\eqref{eq::C0} and ~\eqref{eq::C1} are not entanglement-assisted disjoint\footnote{Two channels are entanglement-assisted disjoint if there exists a common input state, possibly entangled with an auxiliary system, for which the supports of the output states have a trivial intersection.}. 
\end{lemma}

Indeed, for every input state, including inputs entangled with an arbitrary auxiliary system, the two output supports have a nonzero intersection. By the finite-query criterion of Ref.~\cite{Duan09}, these channels cannot be perfectly discriminated with a sequential strategy using any finite number of bare queries. The second part of the proof of Theorem~\ref{thm: perfect discrimination 3 copies}is accomplished by constructing an explicit sequential tester that achieves perfect discrimination with three queries to the averaged purification resource.

\section{Consequences of advantages from purification}
\label{sec: gep implies non simulability}

Our results showing advantages of purified resources have interesting consequences to channel discrimination and to the problem of the convertibility of bare channel uses (R1) into the average over all possible purifications (R2). Here we discuss these consequences.

\begin{corollary}[Sequential over parallel strategy advantage in quantum-to-classical tasks]\label{cor::PARvsSEQ}
    Any strict separation 
    \begin{align}
    \begin{split}
        P_{\mathrm{SEQ}}(\mathsf{C}_\mathrm{bare}^{(k)};s)
        &<P_{\mathrm{SEQ}}(\mathsf{C}_\mathrm{avg}^{(m)};s) \\
        &\implies\\
        P_{\mathrm{PAR}}(\mathsf{C}_\mathrm{avg}^{(k)};s)
        &<P_{\mathrm{SEQ}}(\mathsf{C}_\mathrm{avg}^{(m)};s)
    \end{split}
    \end{align}
    for the same quantum-to-classical task, same number of queries for the same values of $k$ and $m$, and same channel ensemble. Moreover, such a strict separation exists.
\end{corollary}

\begin{proof}
    It follows from Proposition~\ref{prop:unk-eq-avg-channel} that for all score functions $s$, all number of queries $k$, and all channel ensembles,
    \begin{align}
        P_{\mathrm{PAR}}(\mathsf{C}_\mathrm{avg}^{(k)};s) = P_{\mathrm{PAR}}(\mathsf{C}_\mathrm{bare}^{(k)};s).
    \end{align}
    From the fact that parallel strategies are a subset of sequential strategies, it follows that 
    \begin{align}
        P_{\mathrm{PAR}}(\mathsf{C}_\mathrm{bare}^{(k)};s) \leq P_{\mathrm{SEQ}}(\mathsf{C}_\mathrm{bare}^{(k)};s).
    \end{align}
    Now, if there exists a quantum-to-classical task such that for some $s$, $k$, $m$, and channel ensemble
    \begin{align}\label{eq::stricthyp}
        P_{\mathrm{SEQ}}(\mathsf{C}_\mathrm{bare}^{(k)};s) < P_{\mathrm{SEQ}}(\mathsf{C}_\mathrm{avg}^{(m)};s),
    \end{align}
    the above chain of inequalities implies that 
    \begin{align}\label{eq::strictcons}
        P_{\mathrm{PAR}}(\mathsf{C}_\mathrm{avg}^{(k)};s) &< P_{\mathrm{SEQ}}(\mathsf{C}_\mathrm{avg}^{(m)};s).
    \end{align}

    The existence of such an example satisfying the strict inequality in Eq.~\eqref{eq::stricthyp} is guaranteed by Theorem~\ref{thm: perfect discrimination 3 copies}, in particular, where $m=3$ and $k$ is any natural number. Hence, there also exists such an example satisfying the strict inequality in Eq.~\eqref{eq::strictcons}.
\end{proof}

The resource $R^{(k)}_\mathrm{avg}(C)$ is an example of a non-signaling channel. Therefore, this corollary implies the advantage of parallel versus sequential strategies in the discrimination of nonsignalling channels. See Fig.~\ref{fig:plotPARvsSEQvsSEQ} for a numerical plot of the probabilities of success in the example of Theorem~\ref{thm:explicit-2copies-gap}. 
\\

\begin{figure}
\begin{center}
     \includegraphics[width=\columnwidth]{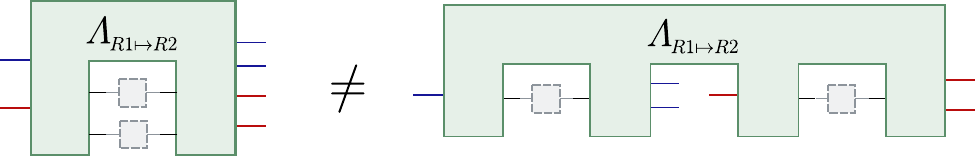}
     \caption{\textbf{Implementation constraints of the transformation converting $\mathrm{R1}$ into $\mathrm{R2}$ (Corollary~\ref{cor::noSEQpurificationsuperchannel}).} The transformation $\Lambda_{\mathrm{R1}\mapsto\mathrm{R2}}$ that can universally transform bare queries to a quantum channel $R^{(k)}_\mathrm{bare}(C)$ into the average of its purifications $R^{(m)}_\mathrm{avg}(C)$ can be implemented by a quantum parallel circuit (left) but not as a sequential quantum circuit that allows for arbitrary interventions in between queries. The image illustrates the case where $k=m=2$.}
     \label{fig:avge purif superchannel}
\end{center}
\end{figure}

We now show how our main results rule out the possibility of converting bare channel queries into the averaged purification queries via sequential-strategy-preserving transformations. The impossibility of such transformation was first shown in Ref.~\cite{yoshida2026randomdilationsuperchannel}, which established their result as a consequence of the query complexity lower bounds on quantum-circuit simulations of the quantum switch~\cite{Bavaresco2025Switch}. More precisely, they prove the impossibility of converting $k\leq\max\{2,2^n-1\}$ queries to an $n$-qubit channel (meaning $d_I=d_O=2^n$) into two queries to the averaged dilation with a strategy that allows for arbitrary interventions in between queries.

In Corollary~\ref{cor::noSEQpurificationsuperchannel} below, we formalize this concept in terms of quantum tester transformations. This uses the general comparison formalism in Ref.~\cite{Jencova_2021}, in which convertibility under a prescribed class of admissible operations implies an ordering of the optimal success probabilities in the corresponding class of learning or quantum-to-classical tasks. Here, the relevant class is the class $\mathrm{SEQ}$ of sequential strategies. We then use Theorem~\ref{thm:explicit-2copies-gap} to provide a simple proof that two bare queries cannot be converted into two queries to the averaged dilation with a strategy that allows for arbitrary intermediate interventions. Also in Corollary~\ref{cor::noSEQpurificationsuperchannel}, we use Theorem~\ref{thm: perfect discrimination 3 copies} to prove that no finite number of bare channel queries $k\in\mathbb{N}$ allows for universal conversion into three or more averaged-purification queries while preserving sequential strategies. This impossibility already holds for qubit-qubit channels.

\begin{corollary}[Impossibility of conversion of bare channel queries into the averaged purification via sequential-strategy preserving transformations]\label{cor::noSEQpurificationsuperchannel}
    There does not exist a linear, positive, and causally well-defined transformation $\Lambda_{\mathrm{R1}\mapsto\mathrm{R2}}$ (i.e. respects the causal constraints of a quantum circuit~\cite{Chiribella2008quantum,Chiribella2009Networks}) that simultaneously satisfies the following two conditions:
    \begin{enumerate}[label={$(\alph*)$}]
        \item $\Lambda_{\mathrm{R1}\mapsto\mathrm{R2}}\left[{R}_\mathrm{bare}^{(k)}(C)\right] = {R}_\mathrm{avg}^{(m)}(C) \ \ \forall \ C$
        \item $\forall \ \{T^{(m)}_y\}_y\in\mathrm{SEQ}, \quad \left\{\Lambda_{\mathrm{R1}\mapsto\mathrm{R2}}^\dagger\left(T^{(m)}_y\right)\right\}_y\in\mathrm{SEQ}$,
    \end{enumerate}
    if 
    \begin{align}
        m = k = 2 \quad \text{or} \quad
        m \geq 3,\ \forall \ k.
    \end{align}

    Consequently, it is impossible to implement a transformation $\Lambda_{\mathrm{R1}\mapsto\mathrm{R2}}$ that converts $k$ bare channel queries into $m$ averaged purification queries, for any dimension $d_I=d_O\geq2$, with a quantum circuit that accepts arbitrary interventions in between queries, as depicted in Fig.~\ref{fig:avge purif superchannel}, if $m = k = 2$ or $\forall \ k\in\mathbb{N}$ if $m\geq 3$. 
\end{corollary}

\begin{proof}
    The proof goes by contradiction. Assume there exists a map $\Lambda_{\mathrm{R1}\mapsto\mathrm{R2}}$ that satisfies conditions (a) and (b) for the specified values of $k$ and $m$. Then, for all quantum-to-classical tasks, via a data processing inequality it follows that 
    \begin{align}
        P_{\mathrm{SEQ}}(\mathsf{C}_\mathrm{bare}^{(k)};s)
        \geq P_{\mathrm{SEQ}}(\mathsf{C}_\mathrm{avg}^{(m)};s),
    \end{align}
    which contradicts Thms.~\ref{thm:explicit-2copies-gap} and~\ref{thm: perfect discrimination 3 copies}.
\end{proof}

\section{Separation between averaged purification and fixed unknown purification in quantum-to-quantum tasks}
\label{sec:q2q tasks}

In the previous sections, we remarked that our three quantum channel access models perform equally in quantum-to-classical tasks involving parallel strategies (including the case of quantum states) and proved a strict advantage of purified resources over bare channel queries in quantum-to-classical tasks involving sequential strategies. Yet, even in the latter case, the two purified access models---namely, access to several queries to a fixed unknown purification or access to the average over all possible purifications---showcase an equal performance: such a fact is a consequence of the linear dependence of the decision function of quantum-to-classical tasks on the channel resource. It then raises the question: is there a scenario where the two purified models of channel access can be operationally distinguished?

This is the case of some quantum-to-quantum tasks---tasks in which the output is itself a quantum object, and the performance depends on properties of the quantum output, rather than a classical score obtained by measuring the quantum output. Some examples of such tasks are quantum state cloning~\cite{Werner_1998} or unitary inversion~\cite{Yoshida_2023}. Here, we consider the example of a task where a quantum referee provides quantum channels to the players under different access models. In turn, the players must return a quantum channel to the quantum referee, who is always able to check the winning condition: if the returned channel is equal to one of the possible purifications of the input channel, the game is won, otherwise it is lost. Trivially, the player that was provided with access to a fixed unknown purification can win the game in every round. On the other hand, the player that was provided with access to the averaged purification cannot convert it into a fixed purification, since it is not possible to universally convert the average purification resource into a fixed unknown purification~\cite{Kleinmann2006physical,Kleinmann2007purifying,liu2026universalpurificationquantummechanics,deltoro2026probabilisticapproximateuniversalquantum}. Hence, players with access to the average purification resource cannot win the game with probability one in every round. This task serves as an example of a scenario where different purification access models can be operationally distinguished. However, it lacks a concrete motivation. It would hence be interesting to find a quantum-to-quantum task that is operationally well motivated and that can distinguish these two different purification resources. 

A diagram depicting an example of a quantum-to-classical and a quantum-to-quantum tasks under different access models is shown in Fig.~\ref{fig::tasks}.

\section{Conclusion}
\label{sec:discussion}

\begin{figure}
        \centering
        \includegraphics[width=\columnwidth]{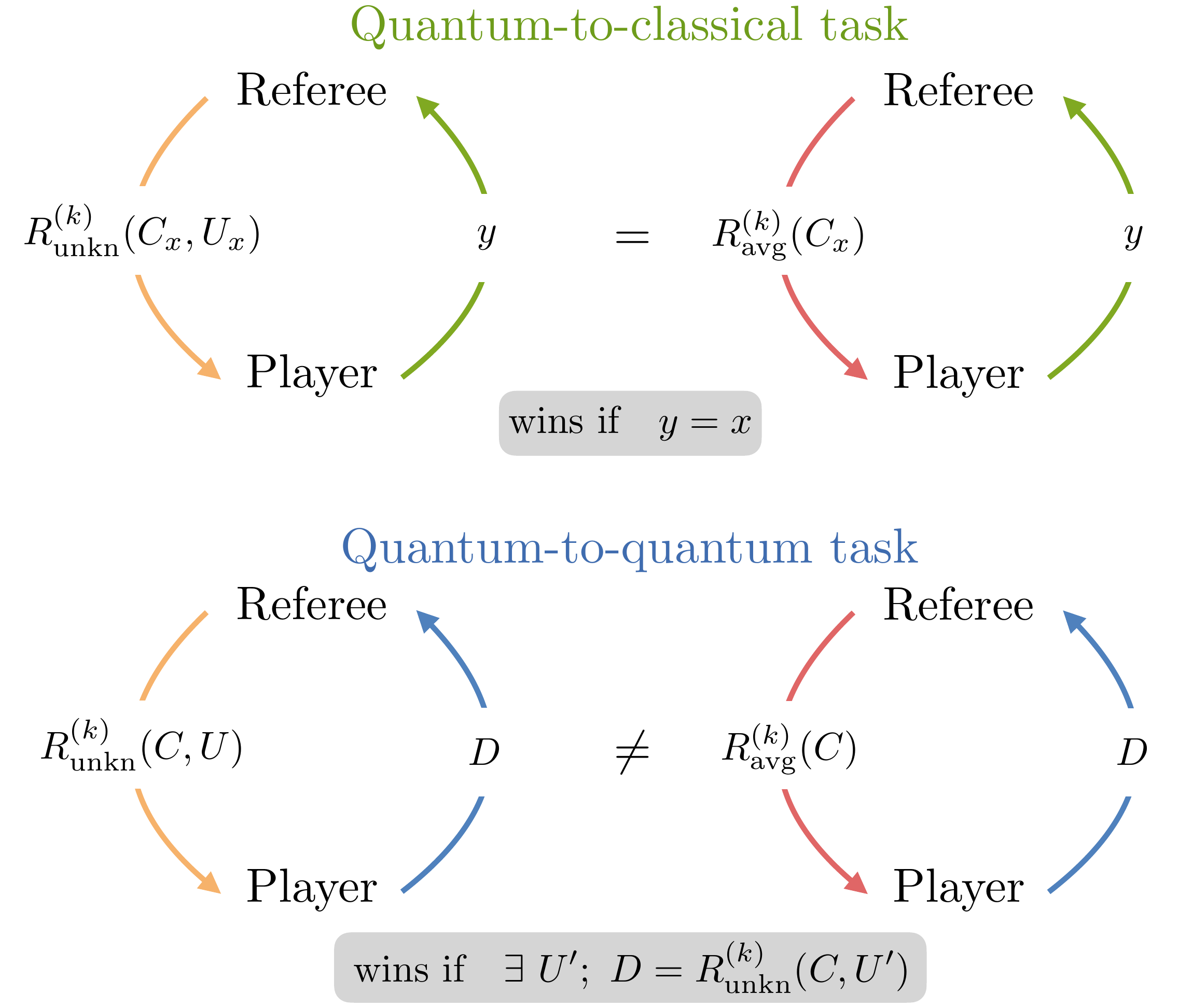}
        \caption{\textbf{Quantum-to-classical and quantum-to-quantum tasks.} Top: Example of a quantum-to-classical task related to discrimination. In all quantum-to-classical tasks, access to a fixed unknown purification (left) and o the averaged purification (right) achieves equal performance. Bottom: Example of a quantum-to-quantum task where access to a fixed unknown purification (left) outperforms access to the averaged purification (right). This task shows how the non-convertibility of these resources can be phrased in operational terms.}
        \label{fig::tasks}
\end{figure}

We have studied the operational value of purification under three models of access to quantum channels: access to bare channel queries, access to averaged purification queries, and access to fixed but unknown purification queries. For quantum-to-classical tasks, the two purified access models are operationally equivalent for both parallel~\cite{Chen2025quantumchannel,chen2025localtestunitarilyinvariant} and sequential strategies. Moreover, in the parallel setting, neither purified access model provides an advantage over bare queries, recovering a result from Refs.~\cite{Chen2025quantumchannel,chen2025localtestunitarilyinvariant}. In contrast, this equivalence breaks down for sequential strategies: we exhibited a pair of qubit-qubit channel families for which two purified queries yield a strictly greater discrimination probability than two bare queries, and another pair that can be perfectly discriminated with three purified queries while remaining impossible to perfectly discriminate with any finite number of bare queries. These separations also imply an advantage of sequential over parallel strategies for the corresponding resources. Moreover, they provide an operational proof that bare channel queries cannot, in general, be converted into averaged purification queries by transformations that allow for arbitrary intermediary interventions.

The present results show that a purification is not merely a mathematical representation of a channel: correlations carried by the purifying systems can have observable operational consequences. A common environment unitary correlates the purified outputs across sequential queries, and Haar averaging organizes these correlations into invariant permutation sectors rather than destroying them. A sequential protocol can retain information from an earlier environment and combine it with a later output. The two-query example exploits the distinct Helstrom contributions of the symmetric and antisymmetric sectors, whereas the three-query construction turns relational information across the purified outputs into an exact support separation. This also explains why our results do not contradict the absence of an advantage in the parallel setting. When bare channel queries are jointly available, a parallel quantum circuit can convert them into the averaged purification resource~\cite{yoshida2026randomdilationsuperchannel,girardi2026randomstinespringsuperchannelconverting}; data processing then precludes any improvement in parallel quantum-to-classical tasks. In the sequential setting, however, the parallel circuit implementing this conversion cannot be combined with arbitrary intermediate interventions without violating its causal constraints. Our discrimination gaps show that this is not merely a limitation of one particular implementation of the conversion, but rather a general obstruction that cannot be circumvented.

At the same time, the equivalence between averaged and fixed unknown purifications in quantum-to-classical tasks does not imply that the two resources are fully interchangeable. We showed that they can be operationally distinguished in quantum-to-quantum tasks. Although our present example is deliberately constructed to expose their difference, this operational comparison clarifies what averaging over the environment does and does not erase. This leaves several natural open questions. In particular, it would be interesting to identify operationally motivated quantum-to-quantum tasks that separate the two purification access models, to characterize which channel families exhibit a purification advantage under sequential strategies for quantum-to-classical tasks.

\vspace*{11pt}
\noindent\textit{Acknowledgments.}
We acknowledge financial support from the French national quantum initiative managed by
Agence Nationale de la Recherche (ANR) in the framework of France 2030 through project EPIQ, with reference ANR-22-PETQ-0007 and the JCJC programme under grant number ANR-25-CE47-6396-01-HOQO-KS.

\vspace*{11pt}
\noindent\textit{AI disclosure.}
All analytical results were derived by the authors, who also created the figures. OpenAI's GPT-5.5 on Codex was used to write the code that produced the numerical plots (see Appendix \ref{ap:: figures}). OpenAI's GPT-6 Astra was used for grammar, spelling checks and text rephrasing. The authors take full responsibility for the content of the work.

\clearpage 
\appendix
\onecolumngrid
\newpage
\section*{APPENDIX}

\appendixtableofcontents
\appsection{Numerical support for the two-query purification advantage (Theorem~\ref{thm:explicit-2copies-gap})}\label{ap:: figures}

We complement the analytical proof of Theorem~\ref{thm:explicit-2copies-gap} and Corollary~\ref{cor::PARvsSEQ} with numerical computations for the pair of parameter-dependent Choi operator families considered therein. 

The discrimination problem considered here admits an exact semidefinite programming (SDP) formulation, which can be solved numerically exactly, i.e., without introducing relaxations or restricting the admissible strategies to a particular ansatz. The computed values therefore represent the global optima up to numerical
precision. Feasible primal and dual solutions provide lower
and upper bounds, respectively, whose difference quantifies
the remaining numerical uncertainty. The analytical proof is independent of these numerical computations and was derived without relying on them. The numerics serve only to illustrate the proven advantage
and assess the tightness of the analytical bounds.

Let $\Hcal_I\simeq\Hcal_O\simeq\mathbb C^2$ and $\Hcal_E\simeq\mathbb C^{d_E}$, with $d_E\geq4$, and let $X,Y,Z$ be the Pauli matrices. On $\Hcal_I\otimes\Hcal_O$, set
\begin{align}
    H\coloneqq X_I\otimes Z_O-Y_I\otimes Y_O-\id_I\otimes X_O,
\end{align}
which defines a family of qubit-qubit channels with Choi operators defined as
\begin{align}
    C_\pm(t)\coloneqq\frac{\id_{IO}}{2}\pm tH
    \in\mathsf L(\Hcal_I\otimes\Hcal_O),
\end{align}
for $0<t<\frac{1}{2\sqrt5}$.

\begin{figure}[H]
\begin{center}
    \includegraphics[width=0.6\textwidth]{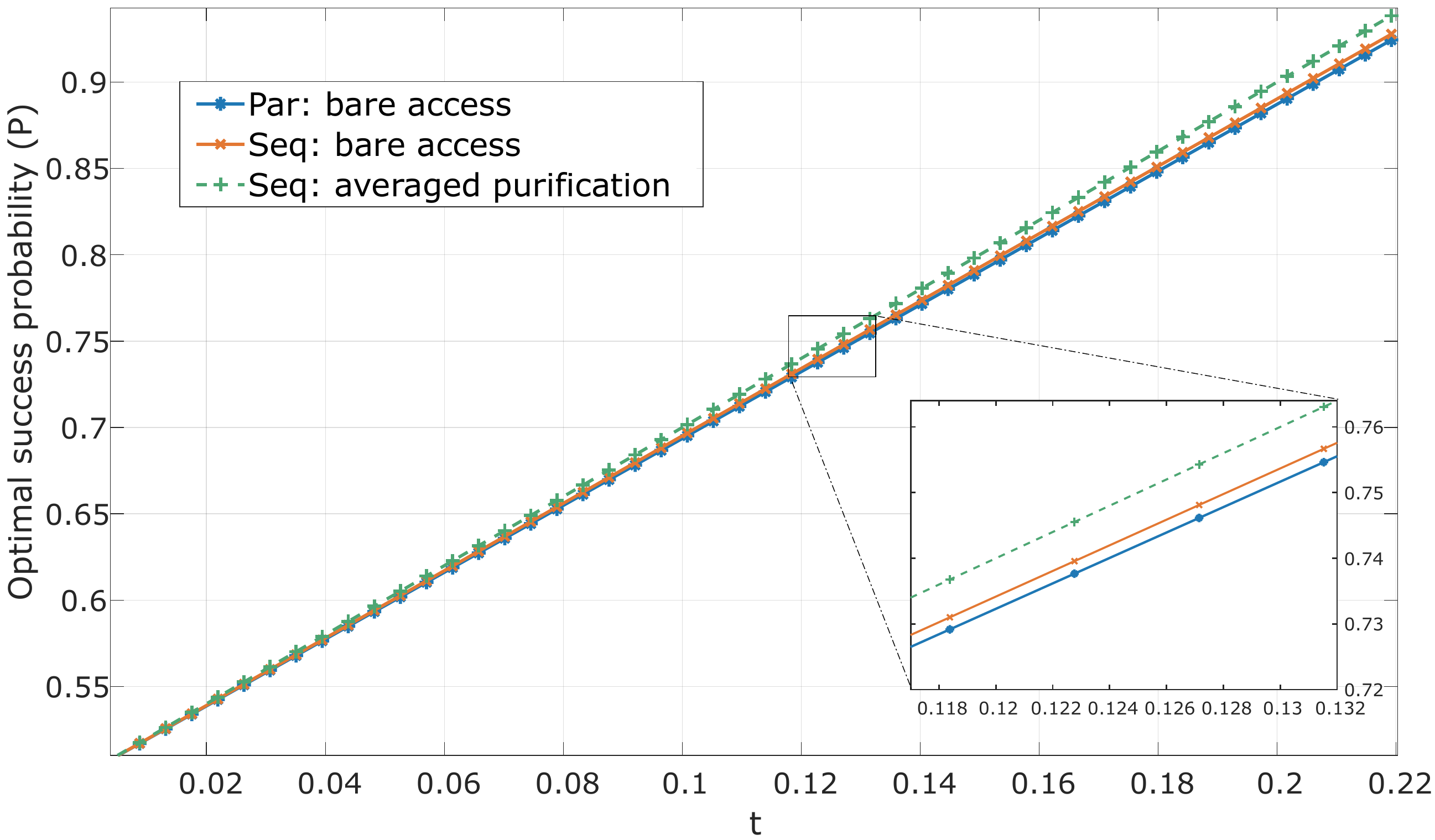}
    \caption{\textbf{Discrimination advantage from purifications implies sequential advantage over parallel strategies (Corollary~\ref{cor::PARvsSEQ}).} Strict advantage in the two-query probability of successful discrimination for the channels in Theorem~\ref{thm:explicit-2copies-gap} under different access models: Bare queries access (R1) under parallel strategies (blue), bare queries access (R1) under sequential strategies (orange), and access to averaged purification (R2) under sequential strategies (green). }
    \label{fig:plotPARvsSEQvsSEQ}
\end{center}        
\end{figure}

Figure~\ref{fig:plotPARvsSEQvsSEQ} illustrates and quantifies the result of Corollary~\ref{cor::PARvsSEQ}, by showing the optimal success probabilities obtained numerically for three strategies: parallel access to bare channel queries, sequential access to the bare channel queries and sequential access to their averaged purifications. All three strategies are evaluated for the same channel pair at each value of the parameter $t$. The comparison displays a gap between sequential access to the averaged purifications and both bare-channel strategies in the regime considered in Theorem~\ref{thm:explicit-2copies-gap}.
It thus illustrates the operational advantage established by the
analytical proof.

\begin{figure}[H]
        \centering
        \includegraphics[width=0.6\textwidth]{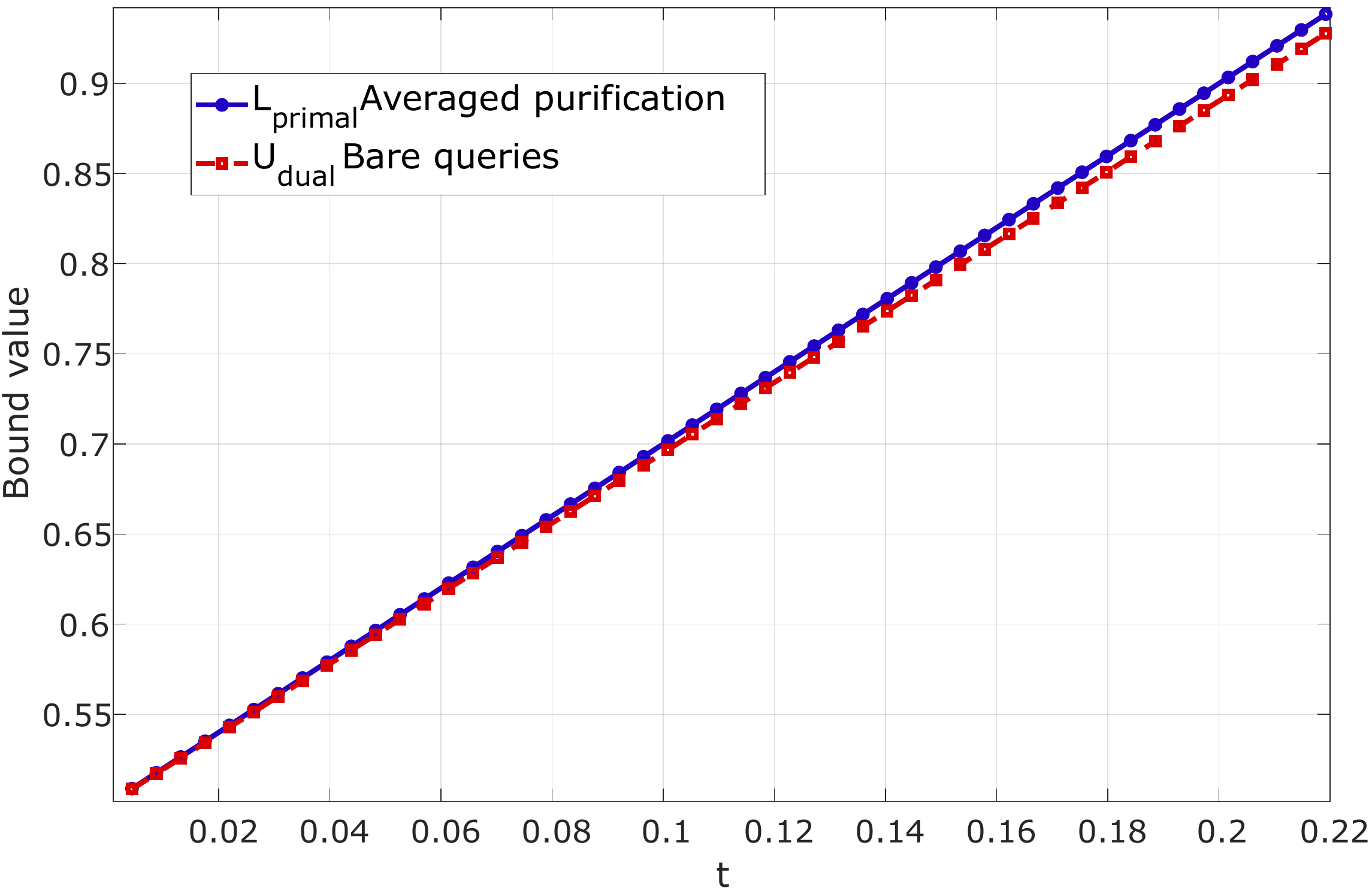}
        \caption{\textbf{Certified gap between two-use discrimination of the averaged purification $R_\mathrm{avg}^{(2)}(C_\pm(t))$ and the bare channels $\{C^{\otimes 2}_\pm(t)\}$}. Ti can b seen that throughout the admissible range of $t$, the primal lower bound for averaged-dilation access (blue) exceeds the dual upper bound for bare access (red). This vertifies the existence of a gap in discrimination both families of channels using the two resources.}
        \label{fig:first}
\end{figure}
To examine this separation more closely,
Fig.~\ref{fig:first} compares the lower and upper bounds
obtained from the primal and dual semidefinite programs,
respectively. A feasible primal solution provides a lower bound
on the optimal success probability, whereas a feasible dual
solution provides an upper bound. In particular, the proof of
Theorem~\ref{thm:explicit-2copies-gap} constructs a primal certificate for sequential discrimination with averaged purifications and a dual certificate for sequential discrimination with bare channels. The separation between the corresponding analytical lower and upper bounds establishes the advantage by weak duality. The numerical primal and dual values also allow us to assess the accuracy of the computed optima through their residual gap.

\begin{figure}[H]
        \centering
        \includegraphics[width=0.6\textwidth]{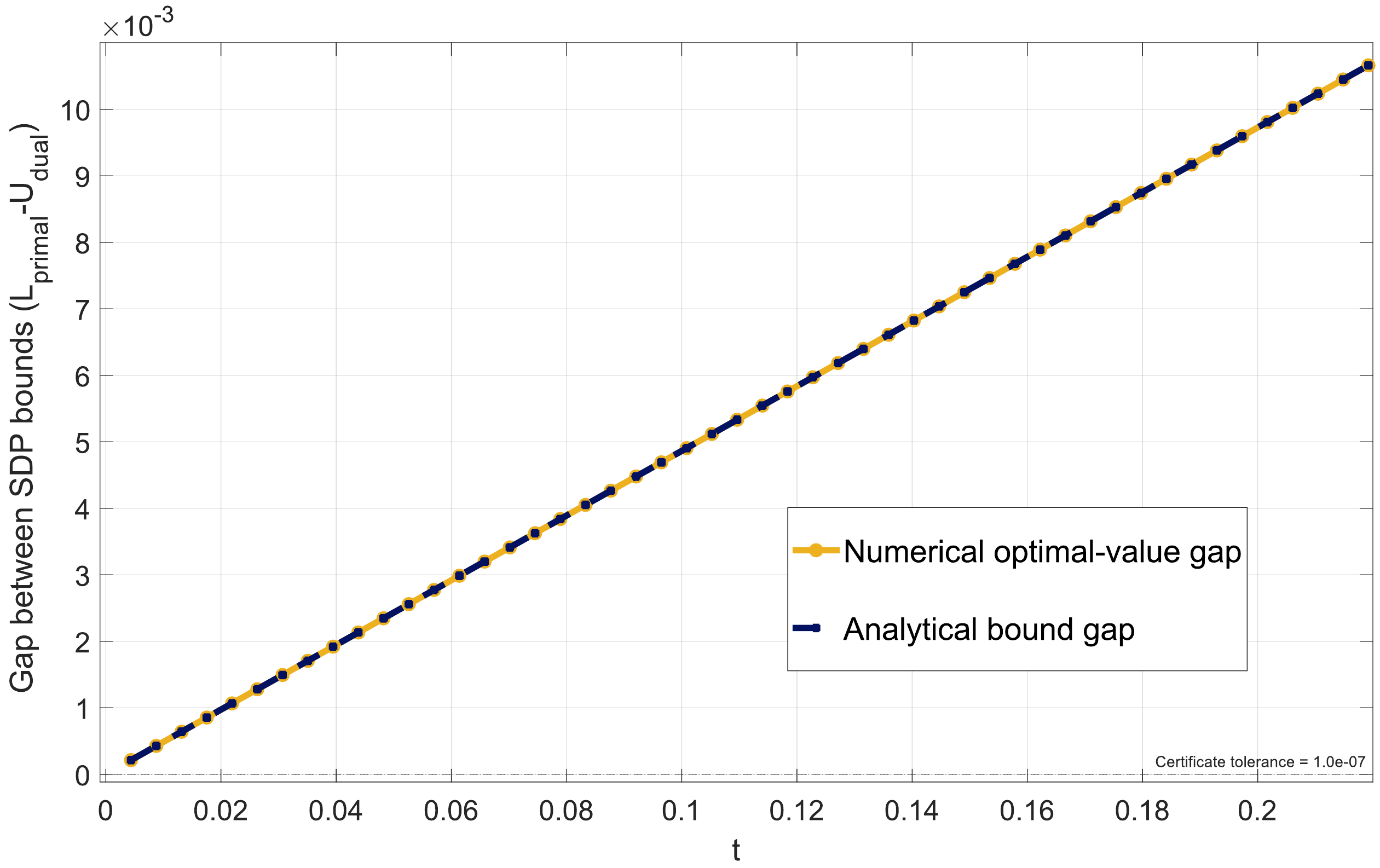}
        \caption{Gap certificate obtained by numerical optimization of the task, compared with the analytical expression.}
        \label{fig:second}
    \end{figure}
Finally, the analytical bounds agree with the corresponding
numerical optima, within numerical precision, at the sampled
parameter values, as shown in Figure \ref{fig:second}. Consequently, the analytical expression for
the guaranteed advantage also reproduces the numerically observed
gap. This agreement provides numerical evidence that the bounds
used in the proof are tight for these channel families. We do not
prove their optimality, since the strict separation between the
analytical bounds already suffices to establish
Theorem~\ref{thm:explicit-2copies-gap}.

\appsection{Weingarten expansion of the averaged purification}\label{ap:: expansion avge purif comb}

This appendix derives an explicit representation of the averaged purification introduced in Sec.~\ref{sec: States, channels, and purification resources }. The calculation is based on the unitary Weingarten calculus and its permutation expansion~\cite{CollinsSniady2006,Collins2022}. We include the relevant steps to make clear how the Haar integral produces the channel-environment correlations used later in the appendix. First, we derive a formula valid for an arbitrary number $k$ of uses. Second, we specialize this formula to $k=2$ and $k=3$, which are the cases used in the proofs of the two-query advantage and the three-query perfect-discrimination result.

Let $\Hcal_A\coloneqq\Hcal_{I}\otimes \Hcal_O$ and let $\Hcal_E$ be the Hilbert space on which the averaged unitary acts. We compute
 \begin{align}
R_\mathrm{avg}^{(k)}(C)=\mathbb E_{U\sim \operatorname{Haar}(\mathsf U(\Hcal_E))}\left[\left((\id_A\otimes U)\ketbra{V_C}(\id_A\otimes U^\dagger)\right)^{\otimes k}\right].
 \end{align}
Equivalently,
 \begin{align}\label{eq: def average purificaiton}
R_\mathrm{avg}^{(k)}(C)=\int_{\mathsf U(\Hcal_E)}\left(\id_{A}^{\otimes k}\otimes U^{\otimes k}\right)\ketbra{V_C}^{\otimes k}\left(\id_{A}^{\otimes k}\otimes (U^\dagger)^{\otimes k}\right)\,\mathrm dU.
 \end{align}
 We now evaluate this Haar integral using unitary Weingarten calculus. Let $\mathsf S_k$ denote the symmetric group of permutations of the
$k$ copies. For $\pi\in \mathsf{S}_k$, let $W_\pi^X$ denote the permutation operator on $\Hcal_X^{\otimes k}$ associated with $\pi$, defined by
 \begin{align}
W_\pi^X(x_1\otimes\cdots\otimes x_k)=x_{\pi^{-1}(1)}\otimes\cdots\otimes x_{\pi^{-1}(k)},\qquad X=A,E.
 \end{align}
 
The Haar integral in Eq.~\eqref{eq: def average purificaiton} averages over a choice of environment reference frame. After this averaging, the resulting operator cannot depend on any particular basis  of $\Hcal_E$. More precisely, this part has to be invariant under every collective unitary transformation $V^{\otimes k}$. By Schur--Weyl duality \cite{CollinsSniady2006,Collins2022}, the operators with this invariance are generated by permutations of the $k$ environment copies. Thus, the Haar integral can be expanded in terms of the operators $W_{\sigma}^{E}$.

The permutation operators are not orthogonal in general. Their overlaps
are given by
\begin{align}
\tr\left(W_\sigma^E W_{\pi^{-1}}^E\right)=d_E^{\#(\sigma\pi^{-1})},
\end{align}
where $\#(\rho)$ denotes the number of cycles of the permutation
$\rho$. The unitary Weingarten function$\operatorname{Wg}_{d_E,k}$ is the inverse of this overlap matrix \cite{CollinsSniady2006,Collins2022}. It therefore determines the coefficients of the permutation expansion. For $d_E\geq k$, this gives

 \begin{align}
\int_{\mathsf U(\Hcal_E)}\left(\id_{A}^{\otimes k}\otimes U^{\otimes k}\right)X\left(\id_{A}^{\otimes k}\otimes (U^\dagger)^{\otimes k}\right)\,\mathrm dU=\sum_{\pi,\sigma\in \mathsf{S}_k}\operatorname{Wg}_{d,k}(\pi^{-1}\sigma)\tr_{E^{\otimes k}} \left[X\left(\id_{A}^{\otimes k}\otimes W_{\pi^{-1}}^E\right)\right]\otimes W_\sigma^E,
 \end{align}
for every
$X\in\mathsf L(\Hcal_A^{\otimes k}\otimes\Hcal_E^{\otimes k})$.
When $d_E<k$, the same expression is understood using the
Moore--Penrose inverse of the permutation Gram matrix.

Substituting into the definition of the average purification in Eq.~\eqref{eq: def average purificaiton} gives
 \begin{align}\label{eq: expanded permutation Scal k}
R_\mathrm{avg}^{(k)}(C)&=\sum_{\pi,\sigma\in \mathsf{S}_k}\operatorname{Wg}_{d_E,k}(\pi^{-1}\sigma)\tr_{E^{\otimes k}}\left[\ketbra{V_C}^{\otimes k}(\id_A\otimes W^{E}_{\pi^{-1}})\right]\otimes W_\sigma^E\\
&= \sum_{\pi,\sigma\in \mathsf{S}_k}\operatorname{Wg}_{d_E,k}(\pi^{-1}\sigma)\tr_{E^{\otimes k}}\left[\ketbra{V_C}^{\otimes k}(W^{A}_{\pi}\otimes \id_E)\right]\otimes W_\sigma^E\\
&= \sum_{\pi,\sigma\in \mathsf{S}_k}\operatorname{Wg}_{d_E,k}(\pi^{-1}\sigma)\left[C^{\otimes k}W^{A}_{\pi}\right]\otimes W_\sigma^E.
 \end{align}

The Weingarten coefficient can be written explicitly as
\begin{align}
\operatorname{Wg}_{d_E,k}(\tau)=
\frac{1}{(k!)^2}\sum_{\substack{\lambda\vdash k\\ \ell(\lambda)\leq d_E}}\frac{(f^\lambda)^2}{s_\lambda(1^{d_E})}\chi^\lambda(\tau),
\end{align}
where $f^\lambda$ is the dimension of the irreducible representation of $\mathsf S_k$ labeled by $\lambda$, $\chi^\lambda$ is its character, and $s_\lambda(1^{d_E})$ is the dimension of the corresponding irreducible representation of $\mathsf U(d_E)$. Equivalently,
\begin{align}
s_\lambda(1^{d_E})=\prod_{(i,j)\in\lambda}\frac{d_E+j-i}{h_{ij}},
\end{align}
where $h_{ij}$ is the hook length of the box $(i,j)$ in the Young diagram of $\lambda$. Hence, Eq.~\eqref{eq: expanded permutation Scal k} is an explicit formula for arbitrary $k$.

\appsubsection{Two-query averaged purification}

We now specialize the general formula of Eq.~\eqref{eq: expanded permutation Scal k} to $k=2$. We assign label $A_i$, $E_i$, $i\in\{1,2\}$ for the first and second copy of the original Hilbert spaces $\Hcal_A$ and $\Hcal_E$. In this case, the group $\mathsf{S}_2$ contains only the identity permutation $e$ and the transposition $(12)$. For $d_E\geq2$, the corresponding Weingarten coefficients are
\begin{align}
\operatorname{Wg}_{d_E,2}(e)&=\frac{1}{d_E^2-1},\\
\operatorname{Wg}_{d_E,2}((12))&=-\frac{1}{d_E(d_E^2-1)}.
\end{align}
Moreover, for $X\in\{A,E\}$, the identity permutation becomes the identity operator $\id$ and the transposition becomes the swap operator $F$ on $\Hcal_X^{\otimes2}$, that is,
\begin{align}
W_e^X=\id_{X_1X_2},\qquad W_{(12)}^X=F_{X_1,X_2}.
\end{align}
Substituting these expressions into Eq.~\eqref{eq: expanded permutation Scal k} gives
\begin{align}
R_\mathrm{avg}^{(2)}(C)&=\frac{1}{d_E^2-1}\Bigg[C^{\otimes2}\otimes\left(\id_{E_1E_2}-\frac{F_{E_1,E_2}}{d_E}\right)+C^{\otimes2}F_{A_1,A_2}\otimes\left(F_{E_1,E_2}-\frac{\id_{E_1E_2}}{d_E}\right)\Bigg].
\label{eq:two-copy-averaged-comb}
\end{align}

To make the representation-theoretic structure explicit, define the symmetric and antisymmetric projectors
\begin{align}
\Pi^{\mathrm{sym}}&\coloneqq\frac{\id^{\otimes 2}+F}{2},\\
\Pi^{\mathrm{asym}}&\coloneqq\frac{\id^{\otimes 2}-F}{2}.
\end{align}
Since $C^{\otimes2}$ commutes with $F_{A_1,A_2}$, Eq.~\eqref{eq:two-copy-averaged-comb} can equivalently be written as
\begin{align}
    R_\mathrm{avg}^{(2)}(C)&=\frac{2}{d_E(d_E+1)}\left[\left(C^{\otimes2}\Pi_{A_1A_2}^{\mathrm{sym}}\right)\otimes \Pi_{E_1E_1}^{\mathrm{sym}}\right]+\frac{2}{d_E(d_E-1)}\left[\left(C^{\otimes2}\Pi^{\mathrm{asym}}_{A_1A_2}\right)\otimes \Pi_{E_1E_2}^{\mathrm{asym}}\right].
\label{eq:two-copy-sector-decomposition}
\end{align}
Thus, the two-copy averaged purification decomposes into a symmetric contribution and an antisymmetric contribution. The same permutation sector appears simultaneously on the channel and environment systems.

Tracing out the purifying systems recovers the two bare channel uses:
\begin{align}
\tr_{E_1E_2}\left[R_\mathrm{avg}^{(2)}(C)\right]=C^{\otimes2}.
\end{align}
For the complementary marginal, using $\tr[C]=d_I$ and
\begin{align}
\tr\left[C^{\otimes2}F_{A_1,A_2}\right]=\tr\left[C^2\right],
\end{align}
we obtain
\begin{align}
    \tr_{A_1A_2}\left[R_\mathrm{avg}^{(2)}(C)\right]&=\frac{1}{d_E^2-1}\Bigg[\left(d_I^2-\frac{\tr[C^2]}{d_E}\right)\id_{E_1E_2}+\left(\tr[C^2]-\frac{d_I^2}{d_E}\right)F_{E_1,E_2}\Bigg] \\
    &=\frac{d_I^2+\tr[C^2]}{d_E(d_E+1)}\, \Pi_{E_1E_2}^{\mathrm{sym}}+\frac{d_I^2-\tr[C^2]}{d_E(d_E-1)}\, \Pi_{E_1E_2}^{\mathrm{asym}}.
\label{eq:two-copy-environment-marginal}
\end{align}

The environmental marginal depends on the channel through
$\tr[C^2]$, which determines the weights of the symmetric and antisymmetric environment sectors. This marginal is useful for identifying the available environment sectors, but it does not by itself capture the discrimination advantage. The advantage depends on the joint channel--environment correlations in Eq.~\eqref{eq:two-copy-sector-decomposition}. A sequential tester can access these correlations, whereas bare channel access discards the environment systems.

The state case is recovered by setting $d_I=1$. The same observation then applies: the environment marginal alone does not determine the discrimination performance.

\appsubsection{Three-query averaged purification}

We now specialize Eq.~\eqref{eq: expanded permutation Scal k} to
$k=3$. The permutations in $\mathsf{S}_3$ split into three conjugacy
classes: the identity, the three transpositions, and the two
three-cycles. We denote the latter two classes by
\begin{align}
\mathsf{T}_3&\coloneqq\{(12),(13),(23)\},\\
\mathsf{K}_3&\coloneqq\{(123),(132)\}.
\end{align}
For $d_E\geq 3$, the corresponding Weingarten coefficients are
\begin{align}
w_e&\coloneqq\operatorname{Wg}_{d_E,3}(e)
=\frac{d_E^2-2}
{d_E(d_E^2-1)(d_E^2-4)},\\
w_t&\coloneqq\operatorname{Wg}_{d_E,3}((12))
=-\frac{1}
{(d_E^2-1)(d_E^2-4)},\\
w_c&\coloneqq\operatorname{Wg}_{d_E,3}((123))
=\frac{2}
{d_E(d_E^2-1)(d_E^2-4)}.
\end{align}

Since the Weingarten coefficient in Eq.~\eqref{eq: expanded permutation Scal k} depends on the relative permutation
$\pi^{-1}\sigma$, we write $\sigma=\pi\tau$ in
Eq.~\eqref{eq: expanded permutation Scal k}. Using the convention
$W_\pi W_\tau=W_{\pi\tau}$, this gives
\begin{align}
R_\mathrm{avg}^{(3)}(C)&= \sum_{\pi\in\mathsf{S}_3}\left(C^{\otimes 3}W_\pi^A\right)\otimes W_\pi^E\left(w_e\id_{E_1E_2E_3}+w_t\sum_{\tau\in\mathsf{T}_3}W_\tau^E+w_c\sum_{\gamma\in\mathsf{K}_3}W_\gamma^E\right).
\label{eq:three-copy-weingarten-central-operator}
\end{align}

For later use, denote the operator in parentheses by
\begin{align}
\Xi_E^{(3)}\coloneqq w_e\id_{E_1E_2E_3}+w_t\sum_{\tau\in\mathsf T_3}W_\tau^E+w_c\sum_{\gamma\in\mathsf C_3}W_\gamma^E .
\end{align}
It is central in the permutation algebra because its coefficient depends only on the conjugacy class of the permutation.

The three-copy environment space admits the Schur--Weyl decomposition
\begin{align}
\Hcal_E^{\otimes 3}\simeq\mathbb{S}_{[3]}(\Hcal_E)\oplus\left(\mathbb{S}_{[2,1]}(\Hcal_E)\otimes\mathbb{C}^{2}\right)\oplus\mathbb{S}_{[1,1,1]}(\Hcal_E).
\end{align}
The labels $[3]$, $[2,1]$, and $[1,1,1]$ correspond, respectively, to the fully symmetric, mixed-symmetry, and fully antisymmetric sectors.
The corresponding central projectors are
\begin{align}
\Pi_E^{[3]}&\coloneqq\frac{1}{6}\left(\id_{E_1E_2E_3}+\sum_{\tau\in\mathsf{T}_3}W_\tau^E+\sum_{\gamma\in\mathsf{K}_3}W_\gamma^E\right),\\
\Pi_E^{[1,1,1]}&\coloneqq\frac{1}{6}\left(\id_{E_1E_2E_3}-\sum_{\tau\in\mathsf{T}_3}W_\tau^E+\sum_{\gamma\in\mathsf{K}_3}W_\gamma^E\right),\\
\Pi_E^{[2,1]}&\coloneqq\id_{E_1E_2E_3}-\Pi_E^{[3]}-\Pi_E^{[1,1,1]}.
\end{align}
We use analogous projectors $\Pi_{A}^{[3]}$, $\Pi_{A}^{[2,1]}$, and $\Pi_{A}^{[1,1,1]}$ on $\Hcal_A^{\otimes 3}$.

In terms of these projectors, the central operator $\Xi_E^{(3)}$ decomposes as
\begin{align}
\Xi_E^{(3)}
&=
\frac{1}{d_E(d_E+1)(d_E+2)}\Pi_E^{[3]}
+\frac{1}{d_E(d_E^2-1)}\Pi_E^{[2,1]}
\nonumber\\
&\quad+
\frac{1}{d_E(d_E-1)(d_E-2)}\Pi_E^{[1,1,1]}.
\label{eq:three-copy-weingarten-projector-decomposition}
\end{align}

The fully symmetric and fully antisymmetric sectors carry one-dimensional irreducible representations of $\mathsf S_3$. The mixed-symmetry sector carries the two-dimensional standard representation. This is why the three-copy expression contains a nontrivial channel-environment correlation that cannot be written simply as a sum of products
$\Pi_{A}^{[\lambda]}\otimes \Pi_E^{[\lambda]}$.

To display this mixed-sector correlation, define
\begin{align}
\Pi_{AE}^{\mathrm{diag},[2,1]}
&\coloneqq
\frac{1}{6}
\sum_{\pi\in\mathsf{S}_3}
\left(\Pi_{A}^{[2,1]}W_\pi^A\right)
\otimes
\left(\Pi_E^{[2,1]}W_\pi^E\right).
\label{eq:diagonal-mixed-projector}
\end{align}

Choosing the same real orthonormal basis $\{\ket{1},\ket{2}\}$ for the two standard-representation multiplicity spaces, this operator is the projector onto the invariant vector
\begin{align}
\ket{\Omega_{AE}}&\coloneqq\frac{1}{\sqrt{2}}\left(\ket{1}_A\ket{1}_E+\ket{2}_A\ket{2}_E\right),
\end{align}
More precisely,
\begin{align}
\Pi_{AE}^{\mathrm{diag},[2,1]}&=\id_{\mathbb S_{[2,1]}(\Hcal_A)}\otimes\id_{\mathbb S_{[2,1]}(\Hcal_E)}\otimes \ketbra{\Omega_{AE}}{\Omega_{AE}} .
\end{align}

Consequently, the full three-copy averaged purification comb takes the
sector-resolved form
\begin{align}
    R_\mathrm{avg}^{(3)}(C)
    &=
    \frac{6}{d_E(d_E+1)(d_E+2)}
    \left(C^{\otimes 3}\Pi_{A}^{[3]}\right)
    \otimes \Pi_E^{[3]}
    \nonumber\\
    &\quad+
    \frac{6}{d_E(d_E^2-1)}
    \left(C^{\otimes 3}\otimes\id_E\right)
    \Pi_{AE}^{\mathrm{diag},[2,1]}
    \nonumber\\
    &\quad+
    \frac{6}{d_E(d_E-1)(d_E-2)}
    \left(C^{\otimes 3}\Pi_{A}^{[1,1,1]}\right)
    \otimes \Pi_E^{[1,1,1]}.
\label{eq:three-copy-sector-decomposition}
\end{align}
Unlike the two-copy case, the three-copy averaged purification cannot be reduced to a sum involving only products of matching channel and environment projectors. The mixed-symmetry contribution retains correlated permutation information between the two parts.

For the perfect-discrimination problem in Thm.~\ref{thm: perfect discrimination 3 copies}, it is useful to characterize the support of the averaged comb directly from its Haar representation.

\setcounter{lemma}{2}

\begin{lemma}[Support of the averaged purification three-comb]
\label{prop:support-averaged-purification-three}
Let $C$ be the Choi operator of a quantum channel, let $r\coloneqq\operatorname{rank}C$, and let $\Hcal_E$ be an auxiliary space with $d_E\geq r$. Then
\begin{align}
\operatorname{supp}R_\mathrm{avg}^{(3)}(C)=\operatorname{Sym}^{3}\left(\operatorname{supp}C\otimes\Hcal_E\right).
\end{align}
\end{lemma}

\begin{proof}
Using the Haar representation, take
\begin{align}
R_\mathrm{avg}^{(3)}(C)&=\int_{\mathsf{U}(\Hcal_E)}
\ketbra{V_{C,U}}^{\otimes 3}\,\mathrm{d}U,\\
\ket{V_{C,U}}&=\left(\id_A\otimes U\right)\ket{V_C}.
\end{align}
Because this is a positive average of rank-one operators, its support is the linear span of the vectors appearing in the average:
\begin{align}
\operatorname{supp}R_\mathrm{avg}^{(3)}(C)=\operatorname{span}\left\{\ket{V_{C,U}}^{\otimes 3}:U\in\mathsf{U}(\Hcal_E)
\right\}.
\label{eq:Haar-orbit-support}
\end{align}
Every vector $\ket{V_C(U)}^{\otimes3}$ is symmetric under permutations of the three copies. Moreover, the channel component of $\ket{V_C}$ belongs to $\operatorname{supp}C$. Hence every vector in the Haar orbit belongs to
\begin{align}
\operatorname{Sym}^{3}\left(\operatorname{supp}C\otimes\Hcal_E\right),
\end{align}
which proves one inclusion.

To prove the converse inclusion, write a Schmidt decomposition of the purification:
\begin{align}
\ket{V_C}=\sum_{\alpha=1}^{r}\sqrt{\lambda_\alpha}\,\ket{c_\alpha}\ket{e_\alpha},\qquad\lambda_\alpha>0.
\end{align}
The degree-three Cauchy decomposition gives
\begin{align}
\operatorname{Sym}^{3}
\left(\operatorname{supp}C\otimes\Hcal_E\right)\simeq\bigoplus_{\substack{\mu\vdash 3\\
\ell(\mu)\leq r}}\mathbb{S}_\mu(\operatorname{supp}C)\otimes\mathbb{S}_\mu(\Hcal_E).
\label{eq:Cauchy-decomposition-three}
\end{align}
The positivity of all Schmidt coefficients implies that the projection
of $\ket{V_C}^{\otimes 3}$ onto each allowed summand has full Schmidt support. The collective unitary action on the environment spans the corresponding representation space $\mathbb{S}_\mu(\Hcal_E)$. Therefore, the Haar orbit spans every summand in
Eq.~\eqref{eq:Cauchy-decomposition-three}. This proves the reverse inclusion.
\end{proof}

\appsection{Symmetry reductions for the discrimination tasks}\label{ap:: SDPs}

This appendix derives the two-copy sequential-discrimination formulas
used in the main text. We consider a binary channel ensemble with equal
prior probabilities
and use the identification score
\begin{align}
    s(x,y)=\delta_{x,y}.
\end{align}

We compare two different two-use resources associated with this same
ensemble. For bare channel access, the two-copy ensemble is
\begin{align}
    \mathsf C_{\mathrm{bare}}^{(2)} = \left\{\left(\frac{1}{2},R_\mathrm{bare}^{(2)}(C_0)\right),\left(\frac{1}{2},R_\mathrm{bare}^{(2)}(C_1)\right)\right\} = \left\{\left(\frac{1}{2},C_0^{\otimes2}\right),\left(\frac{1}{2},C_1^{\otimes2}\right)\right\}.
\end{align}
For averaged purification access, the corresponding ensemble is
\begin{align}
\mathsf C_{\mathrm{avg}}^{(2)}=\left\{\left(\frac{1}{2},R_\mathrm{avg}^{(2)}(C_0)\right),\left(\frac{1}{2},R_\mathrm{avg}^{(2)}(C_1)\right)\right\},
\label{def: 2copies average set}\end{align}
where $R_\mathrm{avg}^{(2)}(C_x)$ denotes the two-copy averaged
purification comb associated with $C_x$, as defined in
Eq.~\eqref{eq::R2} and expanded in Appendix~\ref{ap:: expansion avge purif comb}.
The notation $\mathsf C_{\mathrm{bare}}^{(2)}$ and
$\mathsf C_{\mathrm{avg}}^{(2)}$ refers to the two-use resource
ensembles being discriminated; the underlying hypotheses remain the same
channels $C_0$ and $C_1$.

The aim of this appendix is to derive two expressions for the
corresponding success probabilities. For bare access, we show that
$P_{\mathrm{SEQ}}
\left(
\mathsf C_{\mathrm{bare}}^{(2)};s
\right)$
is obtained by filtering the two-copy difference
\begin{align}
D\coloneqq C_0^{\otimes2}-C_1^{\otimes2}
\end{align}
with a valid deterministic sequential tester and then taking a trace
norm. For averaged purification access, the two environment copies
split into orthogonal symmetric and antisymmetric sectors. We show that
the same tester normalization must be used in both sectors, while the
corresponding filtered channel differences contribute two separate trace
norms.

The following propositions allow us to write the maximal success probability of discriminating a pair of channels with sequential strategies via a maximization over the one-norm. Concept in close relation to the distance between quantum combs~\cite{Chiribella2009Networks} and distance between quantum strategies~\cite{Gutoski_2012}. For completeness, here we present an independent proof with our notation to the trace-norm reduction found in \cite{Chiribella2008Memory}. 

We consider valid deterministic two-round sequential testers throughout.
The tester formalism and its causal normalization conditions are standard
in the theory of higher-order quantum operations
\cite{Chiribella2008Memory,Chiribella2009Networks,taranto2026aguide}.

\appsubsection{Primal SDP}\label{aub ap::  two copies primal SDPs}
We first treat discrimination of two bare channel resources.

Let us begin by defining the set of all operators $T^{(k)}\in\mathsf{L}(\bigotimes_{i=1}^k\Hcal_{I_i}\otimes\Hcal_{\widetilde{O}_i})$, where we recall that $\Hcal_{\widetilde{O}_i}=\Hcal_{O_i}\otimes\Hcal_{E_i}$ such that
\begin{align}
    \{T^{(k)}\in\mathrm{SEQ} \ | \ T^{(k)} = \sum_i T^{(k)}_i, \{T^{(k)}_i\} = \mathfrak{T}^{(k)} \in\mathrm{SEQ} \}, 
\end{align}
that is, $T^{(k)}\in\mathrm{SEQ}$ is the set of all operators that can be resolved into elements of a sequential $k$-query quantum tester. These operators are known as sequential $k$-query quantum combs or sequential $k$-query superchannels. 

The next result concerns the case where all $d_{E_i}=1$, and, hence, $\Hcal_{\widetilde{O}_i}\cong\Hcal_{O_i}$

\begin{proposition}[Sequential two-query discrimination of two quantum channels]
Let $\Ccal_0$ and $\Ccal_1$ be distinct quantum channels with Choi operators $C_i\in \mathsf L(\Hcal_{I}\otimes\Hcal_{O})$, $i\in\{0,1\}$. For a uniform prior distribution, the maximum probability of discriminating them with two sequential queries is
\begin{align}
   P_{\mathrm{SEQ}}(\mathsf{C}_\mathrm{bare}^{(2)};s) &= \frac{1}{2}+\frac{1}{4}\sup_{T^{(2)}\in\mathrm{SEQ}}\norm{\sqrt{T^{(2)}}D\sqrt{T^{(2)}}}_1
\label{eq:reduction primal 2 copies}
\end{align}
where the maximum is over two-query sequential quantum combs $T^{(2)}\in\mathrm{SEQ}$, such that $T^{(2)}\coloneqq T^{(2)}_0+T^{(2)}_1$, where $\mathfrak  T^{(2)}=\{T^{(2)}_0,T^{(2)}_1\} \in \mathrm{SEQ}$ is a sequential quantum tester, and $D\coloneqq C_0^{\otimes 2}-C_1^{\otimes 2}$.
\end{proposition}

\begin{proof}
To simplify the notation of the proof, which concerns the two-query case, we will drop the superscript $(\cdot)^{(2)}$ in the following.

For a uniform prior distribution, the optimal success probability for discriminating both Choi operators with $2$ queries is the solution of the semidefinite program

\begin{align}
    \max\quad & P_{\mathrm{SEQ}}(\mathsf{C}_\mathrm{bare}^{(2)};s)=\frac{1}{2} \tr \!\left[T_0 C_0^{\otimes 2}\right] + \frac{1}{2} \tr \!\left[T_1 C_1^{\otimes 2}\right] \label{eq:raw-primal}\\
    \text{subject to}\quad & T_0 \succeq 0, \qquad T_1 \succeq 0,\label{eq:primal-T-SDP-bare-positivity} \\
    &T= T_0 + T_1,\label{eq:primal-T-SDP-bare-outcome-sum} \\
    &\tr(T)=d_{O_1}d_{O_2}, \label{eq:primal-T-SDP-bare-totaltrace}\\
    &T={}_{O_2}T, \label{eq:primal-T-SDP-bare-normalization}\\
    &{}_{I_2O_2}T=_{O_1I_2O_2}T,\label{eq:causal-normalization-bare}
\end{align}
where we recall that ${}_X Y\coloneqq\tr_X[Y]\otimes \frac{\id_X}{d_X}$.

We define the Hermitian two-copy difference
\begin{align}
D \coloneqq C_0^{\otimes 2} - C_1^{\otimes 2},
\label{eq:general-two-copy-difference}
\end{align}

and the Hermitian decision operator
\begin{align}
G = T_0 - T_1.
\label{eq:raw-decision-operator}
\end{align}

Since $T_0 + T_1 = T$, one has

\begin{align}
T_0 = \frac{T+G}{2}, \qquad T_1 = \frac{T-G}{2}.
\label{eq:tester-outcomes-from-G}
\end{align}

The positivity constraints $T_0, T_1 \succeq 0$ are equivalent to

\begin{align}
-T \preceq G\preceq T.
\label{eq:G-order-interval}
\end{align}

For a uniform prior distribution, substituting into the objective function in Eq.~\eqref{eq:raw-primal} gives

\begin{align}
P_{\mathrm{SEQ}}(\mathsf{C}_\mathrm{bare}^{(2)};s)&=
\frac{1}{2}\Tr\!\left[T_0C_0^{\otimes2}\right]+\frac{1}{2}\Tr\!\left[T_1C_1^{\otimes2}\right]=\frac{1}{2}\Tr\!\left[\left(\frac{T+G}{2}\right)C_0^{\otimes2}\right]+\frac{1}{2}\Tr\!\left[\left(\frac{T-G}{2}\right)C_1^{\otimes2}\right]\\
&=\frac{1}{4}\Tr\!\left[(T+G)C_0^{\otimes2}\right]+\frac{1}{4}\Tr\!\left[(T-G)C_1^{\otimes2}\right]=\frac{1}{4}\Tr\!\left[TC_0^{\otimes2}+GC_0^{\otimes2}+TC_1^{\otimes2}-GC_1^{\otimes2}\right]\\
&=\frac{1}{4}\Tr\!\left[T\left(C_0^{\otimes2}+C_1^{\otimes2}\right)\right]+\frac{1}{4}\Tr\!\left[G\left(C_0^{\otimes2}-C_1^{\otimes2}\right)\right]=\frac{1}{4}\Tr\!\left[T\left(C_0^{\otimes2}+C_1^{\otimes2}\right)\right]+\frac{1}{4}\Tr[GD],
\end{align}
Thus,
\begin{align}
    P_{\mathrm{SEQ}}(\mathsf{C}_\mathrm{bare}^{(2)};s)=\max_{\substack{T\in\mathrm{SEQ}\\-T\preceq G\preceq T}}\left\{\frac14\Tr\!\left[T\left(C_0^{\otimes2}+C_1^{\otimes2}\right)\right]+\frac14\Tr[GD]\right\}.
\end{align}
For every sequential quantum comb $T$, the normalization conditions imply that
\begin{align}
    \tr[TC_i^{\otimes 2}]=1,
\end{align}
for every Choi operator of quantum channel $C_i$.
Therefore, $\tr[T(C_0^{\otimes 2}+C_1^{\otimes 2})]=2$ and the first term in the objective is always $1/2$. The remaining optimization is therefore

\begin{align}
\beta(D)=\max_{T,G}\quad & \Tr[GD] \label{eq:bias-primal-TG}\\
\text{subject to}\quad & -T\preceq G\preceq T, \notag\\
& T\in\mathrm{SEQ}. \notag
\end{align}
We next optimize over $G$ for a fixed $T$. The operator $T$ need not be full rank, so we first restrict the optimization to its support. The constraint $-T\preceq G\preceq T$ is equivalent to $T+G\succeq0$ and $T-G\succeq0$. These inequalities imply \begin{align}
\ker(T)\subseteq\ker(G).
\end{align}

Indeed, if $\ket{\psi}\in\ker(T)$, one has $0\leq\bra{\psi}(T\pm G)\ket{\psi}=\pm\bra{\psi}G\ket{\psi}$, so $\bra{\psi}G\ket{\psi}=0$. Since $T\pm G\succeq0$, this further implies $(T\pm G)\ket{\psi}=0$, and hence $G\ket{\psi}=0$.\\

Let $\Pi_T$ be the projector onto the support of $T$, and let $T^{-1/2}$ denote the Moore--Penrose inverse of $\sqrt{T}$ on this support. Define

\begin{align}
    \widetilde{G}\coloneqq T^{-1/2}GT^{-1/2}
\end{align}
on $\supp(T)$, and set $\widetilde{G}=0$ on $\ker(T)$. Multiplying $T\pm G\succeq0$ on the left and right by $T^{-1/2}$ gives $\Pi_T\pm\widetilde{G}\succeq0$, or equivalently,

\begin{align}
-\Pi_T\preceq \widetilde{G}\preceq\Pi_T.
\label{eq:R-contraction}
\end{align}

Thus, $\widetilde{G}$ is a Hermitian contraction on $\supp(T)$. Conversely, every Hermitian $\widetilde{G}$ satisfying Eq.~\eqref{eq:R-contraction} defines a feasible decision operator

\begin{align}
G=\sqrt{T}\widetilde{G}\sqrt{T}.
\label{eq:W-contraction-parametrization}
\end{align}
because
\begin{align}
T\pm G=\sqrt T\left(\Pi_T\pm\widetilde G\right)\sqrt T\succeq0.
\end{align}
 This operator-interval contraction parametrization follows, for example, from the Douglas factorization lemma~\cite{Douglas1966}; the resulting discrimination expression also appears in Refs.~\cite{Chiribella2008Memory,Chiribella2009Networks}.
Consequently,

\begin{align}
\max_{-T\preceq G\preceq T}\Tr[DG]&=\max_{-\Pi_T\preceq \widetilde{G}\preceq\Pi_T}\Tr\!\left[D\sqrt{T}\widetilde{G}\sqrt{T}\right]\notag\\
&=\max_{-\Pi_T\preceq \widetilde{G}\preceq\Pi_T}\Tr\!\left[\sqrt{T}D\sqrt{T}\,\widetilde{G}\right].
\label{eq:fixed-T-R-reduction}
\end{align}

Define the filtered difference
\begin{align}
    K_T\coloneqq\sqrt{T}D\sqrt{T}
\end{align}
This is the difference between the two hypotheses after they have been processed by a fixed $T$. The remaining optimization is therefore the binary outcome optimization within the class of sequential testers. The filtering by $\sqrt{T}$ records which part of the channel difference can be accessed by that sequential process.

The variational characterization of the trace norm yields

\begin{align}
\max_{-\Pi_T\preceq \widetilde{G}\preceq\Pi_T}\Tr[K_T\widetilde{G}]=\lVert K_T\rVert_1.
\label{eq:trace-norm-variational}
\end{align}

An optimal contraction is
\begin{align}
    \widetilde{G}^\star=\operatorname{sgn}(K_T)=\Pi_+(K_T)-\Pi_-(K_T),
\end{align}
where $\Pi_+(K_T)$ and $\Pi_-(K_T)$ are the spectral projectors onto the positive and negative eigenspaces of $K_T$, respectively. Its value on $\ker(K_T)$ may be chosen arbitrarily in the interval $[-1,1]$. Hence,

\begin{align}
\max_{-T\preceq G\preceq T}\Tr[DG]=\norm{\sqrt{T}D\sqrt{T}}_1,
\label{eq:fixed-T-trace-norm}
\end{align}

Finally, optimizing over all valid deterministic sequential testers
\begin{align}
\beta(D)=\max_{T}\norm{\sqrt{T}D\sqrt{T}}_1.
\label{eq:bias-trace-norm-form}
\end{align}

Combining the preceding steps gives the claimed primal expression:
\begin{align}
    P_{\mathrm{SEQ}}(\mathsf{C}_\mathrm{bare}^{(2)};s)&= \frac{1}{2}+\frac{1}{4}\max_T\norm{\sqrt{T}D\sqrt{T}}_1,\label{eq: en proof prop 1}
\end{align}
where
\begin{align}
    T&\succeq 0 \\
    \tr(T)&=d_{O_1}d_{O_2}, \\
    T&={}_{O_2}T, \\
    {}_{I_2O_2}T&=_{O_1I_2O_2}T.
\end{align}

Notice that the above conditions on $T$ can be equivalently expressed as 
\begin{align}
    N_{I_1 O_1 I_2}\succeq0, &\qquad \rho_{I_1}\succeq0,\\
    \tr\rho_{I_1}&=1 \\
    T&=N_{I_1 O_1 I_2}\otimes\id_{O_2},\\
    \tr_{I_2}N_{I_1 O_1 I_2}&=\rho_{I_1}\otimes\id_{O_1}.
\end{align}
\end{proof}

The same reduction can be applied to the averaged purification ensemble $\mathsf{C}_{\mathrm{avg}}^{(2)}$ in Eq. \eqref{def: 2copies average set}, as we show in the following. The two-copy decomposition derived in Appendix~\ref{ap:: expansion avge purif comb}
allows us to treat the symmetric and antisymmetric environment sectors as
orthogonal flags.

\begin{proposition}[Sequential binary discrimination of two averaged purifications]\label{prop: reduction primal averaged}
For the binary averaged-purification ensemble
\begin{align}
\mathsf{C}_{\mathrm{avg}}^{(2)}=\left\{\left(\frac{1}{2},R_\mathrm{avg}^{(2)}(C_0)\right),\left(\frac{1}{2},R_\mathrm{avg}^{(2)}(C_1)\right)\right\},
\end{align}
the optimal sequential success probability is
\begin{align}
     P_{\mathrm{SEQ}}(\mathsf{C}_\mathrm{avg}^{(2)};s) &= \frac{1}{2}+\frac{1}{4}\max_{T^{(2)}\in\mathrm{SEQ}}\left(\norm{\sqrt{T^{(2)}}\Pi^{\mathrm{sym}}D\Pi^{\mathrm{sym}}\sqrt{T^{(2)}}}_1+\norm{\sqrt{T^{(2)}}\Pi^{\mathrm{asym}}D\Pi^{\mathrm{asym}}\sqrt{T^{(2)}}}_1\right)
\label{eq:reduction-primal-averaged-two-copies}
\end{align}
where the maximum is over two-query sequential quantum combs $T^{(2)}\in\mathrm{SEQ}$, such that $T^{(2)}\coloneqq T^{(2)}_0+T^{(2)}_1$, where $\mathfrak  T^{(2)}=\{T^{(2)}_0,T^{(2)}_1\} \in \mathrm{SEQ}$ is a sequential quantum tester, and $D\coloneqq C_0^{\otimes 2}-C_1^{\otimes 2}$.
\end{proposition}

\begin{proof}
Once again, to simplify the notation of the proof, which concerns the two-query case, we will drop the superscript $(\cdot)^{(2)}$ in the following.

We begin with the primal SDP
\begin{align}
    \max\quad & P_{\mathrm{SEQ}}(\mathsf{C}_\mathrm{avg}^{(2)};s)=\frac{1}{2} \tr \!\left[T_0 R_\mathrm{avg}^{(2)}(C_0)\right] + \frac{1}{2} \tr \!\left[T_1R_\mathrm{avg}^{(2)}(C_1)\right]
    \label{eq:raw-prima-avg}\\
    \text{subject to}\quad & T_0 \succeq 0, \qquad T_1 \succeq 0,\label{eq:primal-T-SDP-avg-positivity} \\
    &T= T_0 + T_1,\label{eq:primal-T-SDP-avg-outcome-sum} \\
    &\tr(T) = d_{O_1}d_{E_1} d_{O_2}d_{E_2},\label{eq:primal-T-SDP-avg-totaltrace} \\
    &T={}_{O_2E_2}T, \label{eq:primal-T-SDP-avg-normalization} \\
    &{}_{I_2O_2E_2}T={}_{O_1E_1I_2O_2E_1}T,\label{eq:causal-normalization-avg}
\end{align}
where we recall once again that ${}_X Y=\tr_X[Y]\otimes\frac{\id_X}{d_X}$.

As shown in Appendix~\ref{ap:: expansion avge purif comb}, Eq.~\eqref{eq:two-copy-sector-decomposition}, for two queries, the averaged purification decomposes into matching symmetric and antisymmetric sectors
\begin{align}
    R_\mathrm{avg}^{(2)}(C)&=\frac{2}{d_E(d_E+1)}\left[\left(C^{\otimes2}\Pi_{IO}^{\mathrm{sym}}\right)\otimes \Pi_E^{\mathrm{sym}}\right]+\frac{2}{d_E(d_E-1)}\left[\left(C^{\otimes2}\Pi_{IO}^{\mathrm{asym}}\right)\otimes \Pi_E^{\mathrm{asym}}\right].
\end{align}

It is convenient to introduce the normalized environment states
\begin{align}
\tau_E^{\mathrm{sym}}&\coloneqq\frac{2\Pi_E^{\mathrm{sym}}}{d_E(d_E+1)},
&\tau_E^{\mathrm{asym}}&\coloneqq\frac{2\Pi_E^{\mathrm{asym}}}{d_E(d_E-1)}.
\end{align}
The two states are supported on orthogonal subspaces because
\[
\Pi_E^{\mathrm{sym}}\Pi_E^{\mathrm{asym}}=0.
\]
They are therefore perfectly distinguishable symmetry-sector flags. Their marginals on the first environment system are nevertheless identical.  Indeed, writing
\begin{align}
\Pi_E^{\mathrm{sym}}=\frac{1}{2}\left(\id_{E_1E_2}+ F_{E_1,E_2}\right),\qquad\Pi_E^{\mathrm{asym}}=\frac{1}{2}\left(\id_{E_1E_2}- F_{E_1,E_2}\right),
\end{align}
and 
\begin{align}
\tr_{E_2}\id_{E_1E_2}=d_E\id_{E_1},\qquad \tr_{E_2}F_{E_1,E_2}=\id_{E_1}.
\end{align}

Consequently,
\begin{align}
\tr_{E_2}\frac{2\Pi_{E}^{\mathrm{sym}}}{d_E(d_E+1)}=\frac{\id_{E_1}}{d_E},\qquad \tr_{E_2}\frac{2\Pi_{E}^{\mathrm{asym}}}{d_E(d_E-1)}=\frac{\id_{E_1}}{d_E}.
\end{align}
This equality is important for the sequential reduction. The two sectors
can be distinguished after the second environment output has been
received, but they induce the same state on the first environment system.
Therefore, both sectors must be compatible with the same sequential
normalization before the second round.

For $j\in\{0,1\}$, $\eta\in\{\mathrm{sym},\mathrm{asym}\}$, define the conditional tester outcomes

\begin{align}
T_{j,\eta}\coloneqq\tr_{E_1E_2}\left[\left(\id_{I_1O_1I_2O_2}\otimes\sqrt{\tau_E^\eta}\right)T_j\left(\id_{I_1O_1I_2O_2}\otimes\sqrt{\tau_E^\eta}\right)\right].
\label{eq:conditional-averaged-tester-outcomes}
\end{align}

Since $T_j\succeq0$, we have

\begin{align}
T_{j,\mathrm{sym}}\succeq0,\qquad T_{j,\mathrm{asym}}\succeq0.
\end{align}

For any operator $X_{IOE}$, any operator $Y_{IO}$, and any state $\tau_E$, the identity

\begin{align}
    \tr\left[X_{IOE}\left(Y_{IO}\otimes\tau_E\right)\right]=\tr\left[Y_{IO}\tr_E\left[\left(\id_{IO}\otimes\sqrt{\tau_E}\right)X_{{IO}E}\left(\id_{IO}\otimes\sqrt{\tau_E}\right)\right]\right]
\end{align}

holds. Applying this identity to the sector decomposition of $R_\mathrm{avg}^{(2)}(C_i)$ gives

\begin{align}
    \tr\left[T_jR_\mathrm{avg}^{(2)}(C_i)\right]=\tr\left[T_{j,\mathrm{sym}}\Pi^{\mathrm{sym}}C_i^{\otimes 2}\Pi^{\mathrm{sym}}\right]+\tr\left[T_{j,\mathrm{asym}}\Pi^{\mathrm{asym}}C_i^{\otimes 2}\Pi^{\mathrm{asym}}\right].\label{eq:sector-resolved-averaged-statistics}
\end{align}

We next determine the normalization of the conditional testers. The
normalization constraints in Eq.~\eqref{eq:primal-T-SDP-avg-normalization} imply the
usual factorization of the final output
\begin{align}
    T=N_{I_1O_1E_1I_2}\otimes\id_{O_2E_2},
\label{eq:extended-tester-factorization}
\end{align}
where $ N_{I_1O_1E_1I_2}$ is the process containing the
first channel output, the environment system $E_1$ retained as memory,
and the input $I_2$ to the second channel use. The causal constraint in
Eq.~\eqref{eq:causal-normalization-avg} implies
\begin{align}
    \tr_{I_2} N_{I_1O_1E_1I_2}=\rho_{I_1}\otimes\id_{O_1E_1}.
\label{eq:extended-causal-normalization}
\end{align}

Using Eq.~\eqref{eq:extended-tester-factorization}, for
$\eta\in\{\mathrm{sym},\mathrm{asym}\}$ we obtain
\begin{align}
T_{0,\eta}+T_{1,\eta}&=\tr_{E_1E_2}\left[\left(\id\otimes\sqrt{\tau_E^\eta}\right) T\left(\id\otimes\sqrt{\tau_E^\eta}\right)\right]\\
&=\tr_{E_1}\left[ N_{I_1O_1E_1I_2}\left(\id_{I_1O_1I_2}\otimes\tr_{E_2}\tau_E^\eta\right)\right]\otimes\id_{O_2}\\
    &=\left(\frac{1}{d_E}\tr_{E_1} N_{I_1O_1E_1I_2}\right)\otimes\id_{O_2}, 
\label{eq:flag-conditioned-normalization}
\end{align}
where we used
\begin{align}
    \tr_{E_2}\tau_E^\eta=\frac{\id_{E_1}}{d_E}
\end{align}
for both $\eta=\mathrm{sym}$ and $\eta=\mathrm{asym}$.


Equation~\eqref{eq:flag-conditioned-normalization} then becomes
\begin{align}
    T_{0,\mathrm{sym}}+T_{1,\mathrm{sym}}=T_{0,\mathrm{asym}}+T_{1,\mathrm{asym}}=T.
\end{align}
Thus $\{T_{0,\mathrm{sym}},T_{1,\mathrm{sym}}\}$ and $\{T_{0,\mathrm{asym}},T_{1,\mathrm{asym}}\}$ are two
conditional binary testers with the same deterministic normalization $T$.
They are not four outcomes of a single tester on the reduced space; indeed,
\begin{align}
    T_{0,\mathrm{sym}}+T_{1,\mathrm{sym}}+T_{0,\mathrm{asym}}+T_{1,\mathrm{asym}}=2T.
\end{align}

We now define the reduced process on the channel systems by
\begin{align}
    N_{I_1O_1I_2}\coloneqq\frac{1}{d_E}\tr_{E_1} N_{I_1O_1E_1I_2}. 
\label{eq:reduced-tester-normalization-averaged}
\end{align}

The reduced operator $N_{I_1O_1I_2}$ satisfies the usual sequential normalization. The reduced process satisfies the usual sequential normalization. Indeed,
using Eq.~\eqref{eq:extended-causal-normalization},

\begin{align}
    \tr_{I_2}N_{I_1O_1I_2}&=\frac{1}{d_E}\tr_{E_1}\left[\tr_{I_2}N_{I_1O_1E_1I_2}\right]=\frac{1}{d_E}\tr_{E_1}\left[\rho_{I_1}\otimes\id_{O_1E_1}\right]=\rho_{I_1}\otimes\id_{O_1}.
\end{align}

Thus, $T$ is a two-query sequential comb of the same form as in the bare-channel problem. Consequently, every feasible tester for the original problem induces the parity-resolved optimization

\begin{align}
P_{\mathrm{SEQ}}\left(\mathsf{C}_{\mathrm{avg}}^{(2)};s\right)=\max\quad &\frac{1}{2}\tr\left[T_{0,\mathrm{sym}}\Pi^{\mathrm{sym}}C_0^{\otimes 2}\Pi^{\mathrm{sym}}\right]+\frac{1}{2}\tr\left[T_{1,\mathrm{sym}}\Pi^{\mathrm{sym}}C_1^{\otimes 2}\Pi^{\mathrm{sym}}\right]\\
&+\frac{1}{2}\tr\left[T_{0,\mathrm{asym}}\Pi^{\mathrm{asym}}C_0^{\otimes 2}\Pi^{\mathrm{asym}}\right]+\frac{1}{2}\tr\left[T_{1,\mathrm{asym}}\Pi^{\mathrm{asym}}C_1^{\otimes 2}\Pi^{\mathrm{asym}}\right]
\label{eq:parity-resolved-binary-sdp}\\
\text{subject to}\quad &T_{0,\mathrm{sym}}\succeq0,\qquad T_{1,\mathrm{sym}}\succeq0,\qquad T_{0,\mathrm{asym}}\succeq0,\qquad T_{1,\mathrm{asym}}\succeq0,\\
&T_{0,\mathrm{sym}}+T_{1,\mathrm{sym}}=N_{I_1O_1I_2}\otimes\id_{O_2},\\
&T_{0,\mathrm{asym}}+T_{1,\mathrm{asym}}=N_{I_1O_1I_2}\otimes\id_{O_2},\\
&N_{I_1O_1I_2}\succeq0,\qquad \rho_{I_1}\succeq0,\\
&\tr_{I_2}N_{I_1O_1I_2}=\rho_{I_1}\otimes\id_{O_1},\\
&\tr\rho_{I_1}=1.
\end{align}

The converse also holds. Given any feasible solution of the
parity-resolved problem, define
\begin{align}
\widehat N_{I_1O_1E_1I_2}\coloneqq N_{I_1O_1I_2}\otimes\id_{E_1},
\end{align}
and, for $j\in\{0,1\}$,
\begin{align}
\widehat T_j\coloneqq T_{j,\mathrm{sym}}\otimes\Pi_E^{\mathrm{sym}}+T_{j,\mathrm{asym}}\otimes\Pi_E^{\mathrm{asym}}.
\end{align}
Since
\[
\Pi_E^{\mathrm{sym}}+\Pi_E^{\mathrm{asym}}=\id_{E_1E_2},
\]
we have
\begin{align}
\widehat T_0+\widehat T_1=\widehat N_{I_1O_1E_1I_2}\otimes\id_{O_2E_2}.
\end{align}
Moreover,
\begin{align}
\tr_{I_2}\widehat N_{I_1O_1E_1I_2}=\rho_{I_1}\otimes\id_{O_1E_1}.
\end{align}
Thus, the reconstructed operators satisfy the original deterministic
tester constraints.

Finally, because the two normalized flag states are supported on the
orthogonal projectors $\Pi_E^{\mathrm{sym}}$ and
$\Pi_E^{\mathrm{asym}}$,
\begin{align}
\tr\left[\widehat T_jR_\mathrm{avg}^{(2)}(C_i)\right]&=\tr\left[T_{j,\mathrm{sym}}\Pi_{IO}^{\mathrm{sym}}C_i^{\otimes2}\Pi_{IO}^{\mathrm{sym}}\right]
\nonumber\\
&\quad+\tr\left[T_{j,\mathrm{asym}}\Pi_{IO}^{\mathrm{asym}}
C_i^{\otimes2}\Pi_{IO}^{\mathrm{asym}}\right].
\end{align}
Therefore, the original and parity-resolved problems have the same
optimal value, and thus, $T=\widehat{T}$, $T_i=\widehat{T}_i$ and $N=\widehat{N}$.

It remains to optimize the two conditional binary measurements for a
fixed deterministic normalization $T$. Define
\begin{align}\label{eq:conditional-tester-common-normalization}
G_{\mathrm{sym}}&\coloneqq T_{0,\mathrm{sym}}-T_{1,\mathrm{sym}},&G_{\mathrm{asym}}&\coloneqq T_{0,\mathrm{asym}}-T_{1,\mathrm{asym}}.
\end{align}
Using Eq.~\eqref{eq:conditional-tester-common-normalization}, we have
\begin{align}
T_{0,\eta}=\frac{T+G_\eta}{2},\qquad T_{1,\eta}=\frac{T-G_\eta}{2},\qquad\eta\in\{\mathrm{sym},\mathrm{asym}\}.
\end{align}
Hence, positivity of the conditional outcomes is equivalent to
\begin{align}
-T\preceq G_{\mathrm{sym}}\preceq T,\qquad-T\preceq G_{\mathrm{asym}}\preceq T.
\label{eq:parity-W-order-interval}
\end{align}

Substituting these expressions into
Eq.~\eqref{eq:parity-resolved-binary-sdp} gives
\begin{align}
P_{\mathrm{SEQ}}\left(\mathsf{C}_{\mathrm{avg}}^{(2)};s\right)&=\frac{1}{4}\sum_{\eta\in\{\mathrm{sym},\mathrm{asym}\}}\tr\left[T\Pi_{IO}^\eta\left(C_0^{\otimes2}+C_1^{\otimes2}\right)\Pi_{IO}^\eta\right]
\nonumber\\
&\quad+\frac{1}{4}\tr\left[G_{\mathrm{sym}}\Pi_{IO}^{\mathrm{sym}}D\Pi_{IO}^{\mathrm{sym}}\right]
+\frac{1}{4}\tr\left[G_{\mathrm{asym}}\Pi_{IO}^{\mathrm{asym}}D\Pi_{IO}^{\mathrm{asym}}\right].
\label{eq:parity-success-probability}
\end{align}

Since
\[
\Pi_{IO}^{\mathrm{sym}}+\Pi_{IO}^{\mathrm{asym}}=\id_{I_1O_1I_2O_2},
\]
and $T$ is a sequential quantum comb, the first term in
Eq.~\eqref{eq:parity-success-probability} equals $1/2$.

For fixed $T$, the two decision operators can be optimized independently.
Applying the trace-norm reduction derived for the bare-channel problem in Eqs.~\eqref{eq:fixed-T-trace-norm}--\eqref{eq: en proof prop 1} to
each symmetry sector yields
\begin{align}
    P_{\mathrm{SEQ}}\left(\mathsf{C}_{\mathrm{avg}}^{(2)};s\right)&=\frac{1}{2}+\frac{1}{4}\max_{T\in\mathrm{SEQ}}\Bigg(\norm{\sqrt T\,\Pi_{IO}^{\mathrm{sym}}D\Pi_{IO}^{\mathrm{sym}}\sqrt T}_1
+\norm{\sqrt T\,\Pi_{IO}^{\mathrm{asym}}D\Pi_{IO}^{\mathrm{asym}}\sqrt T}_1\Bigg).
\end{align}
This proves the proposition.
\end{proof}

\appsubsection{Dual SDP}

For the certification of the two-copy advantage, we only need the dual
problem for the bare-channel ensemble. A feasible point of this dual problem gives an upper bound on the optimal value of $P_{\mathrm{SEQ}}(\mathsf{C}_{\mathrm{bare}}^{(2)};s)$.

The dual problem of the discrimination task in Eq.~\eqref{eq:bias-primal-TG} is
\begin{align}\label{eq: dual 2 bare copies}
\frac{1}{2}+\frac{1}{4}\min_{R,Q,\lambda}\quad
& \lambda \\
\mathrm{s.t.}\quad
& R\succeq D,\qquad R\succeq-D,\\
& \tr_{O_2}(R)\preceq Q_{I_1O_1}\otimes \id_{I_2},\\
& \tr_{O_1}Q_{I_1O_1}\preceq\lambda\id_{I_1},\\
& Q_{I_1O_1}=Q_{I_1O_1}^\dagger,\qquad \lambda\in\mathbb R.
\end{align}

For completion we derive it from the primal expression in Eq.~\eqref{eq:bias-primal-TG}:
\begin{align}
\max_{T,G}\quad & \Tr[GD] \\
\text{subject to}\quad & -T\preceq G\preceq T, \notag\\
& T\in \mathrm{SEQ}. \notag
\end{align}

Although the primal problem above was written using normalized
partial-trace constraints, it is equivalent to the usual causal
parametrization in terms of an intermediate process $N$ and an initial
state $\rho$ such that
\begin{align}
T=N_{I_1O_1I_2}\otimes\id_{O_2},\qquad N_
{I_1O_1I_2}\succeq0,
\label{eq:dual-causal-parametrization-bare}
\end{align}
where $N_{I_1O_1I_2}$ contains the preparation, memory, and input to
the second channel use. The causal normalization is
\begin{align}
\tr_{I_2}N_{I_1O_1I_2}=\rho_{I_1}\otimes\id_{O_1},\qquad\rho_{I_1}\succeq0,
\label{eq:dual-causal-normalization-bare}
\end{align}
Then the optimization task reads
\begin{align}
P_{\mathrm{SEQ}}\left(\mathsf{C}_{\mathrm{bare}}^{(2)};s\right)=\frac{1}{2}+\frac{1}{4}\max_{G,N,\rho}\quad
& \tr[GD] \label{eq:causal-primal-bare}\\
\text{subject to}\quad
& N_{I_1O_1I_2}\otimes \id_{O_2}+G\succeq0,\label{eq:causal-primal-bare-plus}\\
& N_{I_1O_1I_2}\otimes \id_{O_2}-G\succeq0,\label{eq:causal-primal-bare-minus}\\
& N_{I_1O_1I_2}\succeq0,\qquad \rho_{I_1}\succeq0,\label{eq:causal-primal-bare-positivity}\\
& \tr_{I_2}N_{I_1O_1I_2}=\rho_{I_1}\otimes \id_{O_1},\label{eq:causal-primal-bare-causality}\\
& \tr\rho_{I_1}=1.\label{eq:causal-primal-bare-state-normalization}
\end{align}

To construct the dual, introduce positive-semidefinite operators $J\succeq0$ and $K\succeq0$ for the two semidefinite inequalitiesin
Eqs.~\eqref{eq:causal-primal-bare-plus} and\eqref{eq:causal-primal-bare-minus}. Introduce a Hermitian dual variable
\begin{align}
Q_{I_1O_1}=Q_{I_1O_1}^\dagger
\end{align}
 for the causal equality in Eq.~\eqref{eq:causal-primal-bare-causality}, and a real scalar$\lambda\in\mathbb R$ for the normalization condition in Eq.~\eqref{eq:causal-primal-bare-state-normalization}.

The positive operators $J$ and $K$ act as certificates for the two
semidefinite tester constraints. The operator $Q$ enforces the causal
normalization, and $\lambda$ enforces the normalization of the initial
state. The corresponding Lagrangian is
\begin{align}
\mathcal{L}=&\tr[GD]+\tr[J(N\otimes \id_{O_2}+G)]+\tr[K(N\otimes \id_{O_2}-G)]\\
    &-\tr[Q(\tr_{I_2}[N]-\rho\otimes\id_{O_1})]-\lambda(\tr[\rho]-1).
\end{align}
Regrouping terms according to the primal variables gives
\begin{align}
  \mathcal{L}=&\tr[G(D+J-K)]+\tr\!\left[N\left(\tr_{O_2}(J+K)-Q\otimes \id_{I_2}\right)\right]\nonumber\\
    &+\tr\!\left[\rho\left(\tr_{O_1}Q-\lambda \id_{I_1}\right)\right]+\lambda.
\end{align}
The dual function is obtained by taking the supremum of this expression
over the primal variables $G$, $N$, and $\rho$. For this supremum to
be finite, the coefficient of every unbounded primal variable must satisfy
the corresponding dual constraint.

First, $G$ is an unconstrained Hermitian operator. Therefore,
\begin{align}
\sup_{G=G^\dagger}\tr\left[G(D+J-K)\right]
\end{align}
is finite if and only if
\begin{align}
D+J-K=0.
\end{align}
Equivalently,
\begin{align}
K-J=D.
\end{align}

Next, the variable $N$ is positive semidefinite. Hence,
\begin{align}
\sup_{N\succeq0}
\tr\left[N\left(\tr_{O_2}(J+K)-Q\otimes \id_{I_2}\right)\right]
\end{align}
is finite if and only if
\begin{align}
\tr_{O_2}(J+K)-Q\otimes \id_{I_2}\preceq0.
\end{align}
Equivalently,
\begin{align}
\tr_{O_2}(J+K)\preceq Q\otimes \id_{I_2}.
\end{align}
This is the dual counterpart of the causal normalization of the sequential
tester.

Similarly, because $\rho$ is positive semidefinite,
\begin{align}
\sup_{\rho\succeq0}\tr\left[\rho\left(\tr_{O_1}Q-\lambda\id_{I_1}\right)\right]
\end{align}
is finite if and only if 
\begin{align}
\tr_{O_1}Q \preceq \lambda\id_{I_1}.
\end{align}
This constraint is the dual counterpart of the normalization
$\tr\rho_{I_1}=1$.

When all these finiteness conditions hold, all terms depending on $G$, $N$, and $\rho$ are bounded above, and the supremum of the Lagrangian reduces to the constant term:
\begin{align}
\sup_{G,N,\rho}\mathcal{L}=\lambda.
\end{align}

The dual can therefore first be written as
\begin{align}\label{eq:dual-JK-form}
\frac{1}{2}+\frac{1}{4}\min_{J,K,Q,\lambda}\quad & \lambda \\
\text{subject to}\quad
& J\succeq0,\qquad K\succeq0,\\
& K-J=D,\\
& \tr_{O_2}(J+K)\preceq Q_{I_1O_1}\otimes \id_{I_2},\\
& \tr_{O_1}Q_{I_1O_1} \preceq \lambda\id_{I_1},\\
& Q_{I_1O_1}=Q_{I_1O_1}^\dagger,\qquad \lambda\in\mathbb R.
\end{align}

It remains to combine the two positive operators $J$ and $K$ into a
single dual variable. Define
\begin{align}
R\coloneqq J+K.
\end{align}
Using $K-J=D$, we find
\[
R-D=2J,\qquad R+D=2K.
\]
Thus, $J\succeq0$ and $K\succeq0$ are equivalent to
\begin{align}
R\succeq D,\qquad R\succeq-D.
\end{align}
Moreover,
\[
J+K=R.
\]
Substituting these relations into Eq.~\eqref{eq:dual-JK-form} yields
\begin{align}
\frac{1}{2}+\frac{1}{4}\min_{R,Q,\lambda}\quad& \lambda \\
\text{subject to}\quad
& R\succeq D,\qquad R\succeq-D,\\
& \tr_{O_2}(R)\preceq Q_{I_1O_1}\otimes \id_{I_2},\\
& \tr_{O_1}Q_{I_1O_1} \preceq \lambda\id_{I_1},\\
& Q_{I_1O_1}=Q_{I_1O_1}^\dagger,\qquad \lambda\in\mathbb R.
\end{align}
Any feasible choice of $R$, $Q$, and $\lambda$ therefore provides an
upper bound on the optimal bare-channel success probability. This is the
form used below to construct the dual certificate for the two-copy qubit
example.

\appsection{Proof of the two-copy qubit advantage}\label{ap:: qubit advantage}

This appendix proves the strict two-use advantage of averaged-purification
access for a pair of qubit channels (Theorem \ref{thm:explicit-2copies-gap}) in the main text. The proof compares a feasible primal certificate for the averaged-purification discrimination problem with a feasible dual certificate for the bare-channel problem.

\setcounter{theorem}{0}

\begin{theorem}[Two-query sequential channel discrimination advantage from purification]
    In a channel discrimination task, where the score function is $s(x,y)=\delta_{x,y}$, for $k=2$ queries, there exist channel ensembles such that, for sequential strategies,
    \begin{align}
        P_{\mathrm{SEQ}}(\mathsf{C}_\mathrm{bare}^{(2)};s)
        < P_{\mathrm{SEQ}}(\mathsf{C}_\mathrm{avg}^{(2)};s) = P_{\mathrm{SEQ}}(\mathsf{C}_\mathrm{ukn}^{(2)};s).
    \end{align}
    \\
    
    In particular, let $\Hcal_I\simeq\Hcal_O\simeq\mathbb C^2$ and $\Hcal_E\simeq\mathbb C^{d_E}$, with $d_E\geq4$, and let $X,Y,Z$ be the Pauli matrices. On $\Hcal_I\otimes\Hcal_O$, set
    \begin{align}
        H\coloneqq X_I\otimes Z_O-Y_I\otimes Y_O-\id_I\otimes X_O.
    \end{align}
    Now take the hypothesis set $\mathsf{X}=\{+,-\}$ and the decision set $\mathsf{Y}=\mathsf{X}$. Encode the hypothesis set into the family of qubit-qubit channels with Choi operators defined as
    \begin{align}
        C_\pm(t)\coloneqq\frac{\id_{IO}}{2}\pm tH
        \in\mathsf L(\Hcal_I\otimes\Hcal_O),
    \end{align}
    for $0<t<\frac{1}{2\sqrt5}$.
    Then, set the channel ensembles $\mathsf{C}_\mathrm{bare}^{(2)}(t) \coloneqq \{p(\pm),R^{(2)}_\mathrm{bare}(C_\pm(t))\}$ and $\mathsf{C}_\mathrm{avg}^{(2)}(t) \coloneqq \{p(\pm),R^{(2)}_\mathrm{avg}(C_\pm(t))\}$   
    according to the two access models in Eqs.~\eqref{eq::R1} and~\eqref{eq::R2} and by taking an equal prior distribution $p(+)=p(-)=1/2$. Then, considering sequential strategies,
    \begin{align}
        P_{\mathrm{SEQ}}(\mathsf{C}_\mathrm{avg}^{(2)}(t);s)
        -P_{\mathrm{SEQ}}(\mathsf{C}_\mathrm{bare}^{(2)}(t);s)
        \geq \frac{7-3\sqrt{5}}{6}\,t>0.
    \end{align}
\end{theorem}

\begin{proof}
We prove the separation by constructing two explicit semidefinite-programming
certificates. Proposition~\ref{prop: reduction primal averaged} gives the
primal formulation for sequential discrimination of the averaged
purification ensemble, while equation~\ref{eq: dual 2 bare copies}
gives the dual formulation for the corresponding bare-channel ensemble.
Both are finite-dimensional SDPs satisfying the appropriate Slater
conditions, so strong duality holds. The argument below, however, only
requires weak duality.

A feasible deterministic tester normalization $T^\star$ in the primal
problem gives an achievable success probability and therefore a lower
bound
\[
L_{\mathrm{primal}}(t)\leq P_{\mathrm{SEQ}}\left(\mathsf C_{\mathrm{avg}}^{(2)}(t);s\right).
\]
Likewise, any feasible dual point
$(R^\star,Q^\star,\lambda^\star)$ in the dual problem gives an upperbound
\[
P_{\mathrm{SEQ}}\left(\mathsf C_{\mathrm{bare}}^{(2)}(t);s\right)\leq U_{\mathrm{dual}}(t).
\]
It is therefore sufficient to verify the positivity and causal
constraints of the two certificates and to show that
\begin{align}
    U_{\mathrm{dual}}(t)<L_{\mathrm{primal}}(t).
\end{align}

The proof has three steps. We first determine the Pauli-sector structure
of the perturbation $H$. We then construct a feasible primal tester
normalization for the averaged-purification problem and evaluate the
corresponding symmetry-resolved objective. Finally, we construct a
feasible dual point for the bare-channel problem and compare the resulting
upper and lower bounds.\\


We begin by isolating the Pauli strings that generate the algebra used in both constructions. Introduce the three operators

\begin{align}
A_1&=\id_I\otimes X_O,\\ 
A_2&=X_I\otimes Z_O-Y_I\otimes Y_O,\\
A_3&=Z_I\otimes X_O. \label{eq:explicit-gap-generators} 
\end{align}

Then $H=A_2-A_1$. Direct Pauli multiplication gives

\begin{align} 
A_1^2=A_3^2=\id_{IO},\qquad [A_1,A_3]=[A_2,A_3]=\{A_1,A_2\}=0,\qquad A_2^2=2(\id_{IO}-A_3). 
\end{align}

Consequently,

\begin{align} 
H^2=3\id_{IO}-2A_3. 
\end{align}

In particular, $A_3$ commutes with $H$, so $H$ is block diagonal with respect to the projectors

\begin{align}
 P_\pm\coloneqq\frac{\id_{IO}\pm Z_I\otimes X_O}{2},\qquad \Hcal_I\otimes\Hcal_O=\operatorname{ran}\Pi_+\oplus\operatorname{ran}\Pi_-.
\label{eq:explicit-gap-sector-decomposition}
\end{align}

On the two sectors,

\begin{align} 
H^2\Pi_+=\Pi_+,\qquad H^2\Pi_-=5\Pi_-. 
\end{align}

Therefore,
\begin{align} 
\norm{H}_\infty=\sqrt{5}. 
\end{align}

Moreover, $\tr_OH=0$, so that $\tr_OC_\pm(t)=\id_I$, and

\begin{align} 
C_\pm(t)\succeq\left(\frac{1}{2}-\sqrt{5}\,t\right)\id_{IO}\succ0 
\end{align}

whenever $0<t<1/(2\sqrt{5})$. Thus $C_\pm(t)$ are full-rank valid channel Choi operators in this range.

The two-copy difference is linear in $t$. Define the $t$-independent two-copy difference

\begin{align} 
D&\coloneqq\frac{1}{t}\left(C_+(t)^{\otimes2}-C_-(t)^{\otimes2}\right)\\
&=\id_{I_1O_1}\otimes H_{I_2O_2}+H_{I_1O_1}\otimes\id_{I_2O_2}. \label{eq:explicit-gap-scaled-difference} 
\end{align}

The terms quadratic on $t$ because the two hypotheses differ only by the sign of the perturbation $tH$.

By the averaged-purification discrimination reduction in Proposition~\ref{prop: reduction primal averaged},

\begin{align} 
P_{\mathrm{SEQ}}\left(\mathsf C_{\mathrm{avg}}^{(2)}(t);s\right)=\frac{1}{2}+\frac{t}{4}\max_T\left\{\norm{\sqrt{T}\Pi^{\mathrm{sym}}D\Pi^{\mathrm{sym}}\sqrt{T}}_1+\norm{\sqrt{T}\Pi^{\mathrm{asym}}D\Pi^{\mathrm{asym}}\sqrt{T}}_1\right\}.
\label{eq:explicit-gap-scal-primal} 
\end{align}

\appsubsection{Primal certificate}

We now construct a feasible deterministic sequential tester normalization for the averaged-purification problem. The symmetry reduction in Eq.~\eqref{eq:parity-resolved-binary-sdp} shows that the tester should be chosen to make both the symmetric and antisymmetric channel sectors contribute constructively to the trace-norm objective.\\

The starting point is the maximally mixed deterministic normalization $N_0=\frac{\id_{I_1O_1I_2}}{4}$, whose causal marginal is $\tr_{I_2}N_0=\frac{\id_{I_1}}{2}\otimes\id_{O_1}$. To identify useful perturbations of $N_0$ that can give rise to a valid tester, we use the local algebra of $H$. On the $A_3=+1$ sector, $H\Pi_+=-(\id\otimes X)\Pi_+$, while $\Pi_+(\id\otimes X)=\left(\id\otimes X+Z\otimes\id\right)/2$. This singles out the two first-copy directions $\id_{I_1}\otimes X_{O_1}$ and $Z_{I_1}\otimes\id_{O_1}$. To preserve the causal marginal, these directions must be coupled to an operator with vanishing partial trace on $\Hcal_{I_2}$; the simplest Pauli choice is $Z_{I_2}$. This leads to the restricted family $N_{r,s}=N_0+r\,\id_{I_1}\otimes X_{O_1}\otimes Z_{I_2}+s\,Z_{I_1}\otimes\id_{O_1}\otimes Z_{I_2}$. Positivity restricts the coefficients to $|r|+|s|\leq1/4$, and for which we fix $r=1/6$ and $s=1/12$. Thus the coefficients in the final ansatz are calibrated by the primal objective under positivity and causality, rather than chosen independently.

The resulting feasible primal point is
  \begin{align}
      T^\star&= N^\star \otimes \id_{O_2}\\
      N^\star&=\left(\frac{1}{4}\id_{I_1O_1I_2}+\frac{1}{6}\id_{I_1}\otimes X_{O_1}\otimes Z_{I_2}+\frac{1}{12}Z_{I_1}\otimes\id_{O_1}\otimes Z_{I_2}\right)
  \end{align}

We first verify positivity. Both Pauli strings are Hermitian involutions, so the operators

\begin{align}
\frac{1}{2}\left(\id_{I_1O_1I_2}+\id_{I_1}\otimes X_{O_1}\otimes Z_{I_2}\right),\qquad \frac{1}{2}\left(\id_{I_1O_1I_2}+Z_{I_1}\otimes\id_{O_1}\otimes Z_{I_2}\right) 
\end{align}
are projectors. Moreover,
\begin{align}
N^\star=\frac13\,\frac{\id+\id_{I_1}\otimes X_{O_1}\otimes Z_{I_2}}{2}+\frac16\,\frac{\id+Z_{I_1}\otimes\id_{O_1}\otimes Z_{I_2}}{2}.
\end{align}
Thus $N^\star$, and hence $T^\star$, is positive semidefinite:
 \begin{align}
     T^\star\succeq 0.
 \end{align}

Both perturbation terms in $N^\star$ contain the traceless operator $Z_{I_2}$. Their partial traces over $I_2$ therefore vanish, and

\begin{align} \tr_{I_2}N^\star=\frac{1}{4}\tr_{I_2}\id_{I_1O_1I_2}=\frac{1}{2}\id_{I_1O_1}=\frac{\id_{I_1}}{2}\otimes\id_{O_1}. \end{align}

Defining

\begin{align} \rho_{I_1}\coloneqq\frac{\id_{I_1}}{2}, \end{align}

we have

\begin{align} \rho_{I_1}\succeq0,\qquad \tr\rho_{I_1}=1,\qquad \tr_{I_2}N^\star=\rho_{I_1}\otimes\id_{O_1}. \end{align}

Finally, $T^\star$  has the required final-output factorization

\begin{align} T^\star=N^\star\otimes\id_{O_2}, 
\end{align}

and its total normalization is

\begin{align} \tr T^\star=\tr N^\star\,\tr\id_{O_2}=d_{O_1}d_{O_2}=4. \end{align}

Therefore $T^\star$ satisfies positivity and all the sequential causality and normalization constraints, and is a valid deterministic two-step sequential tester normalization.\\

It remains to evaluate the objective achieved by $T^\star$. The operator $D$ is invariant under exchanging the two copies, and therefore commutes with the two-copy projectors $\Pi_{IO}^{\mathrm{sym}}$ and $\Pi_{IO}^{\mathrm{asym}}$. Using $\det(\lambda\id-AB)=\det(\lambda\id-BA)$, we obtain
\begin{align} 
\det(\lambda\id-\sqrt{T}\Pi^{\mathrm{sym}}D\Pi^{\mathrm{sym}}\sqrt{T})&=\det\left(\lambda\id-T^\star D\Pi^{\mathrm{sym}}\right),\\
\det(\lambda\id-\sqrt{T}\Pi^{\mathrm{asym}}D\Pi^{\mathrm{asym}}\sqrt{T})&=\det\left(\lambda\id-T^\star D\Pi^{\mathrm{asym}}\right). \label{eq:primal-characteristic-reduction} 
\end{align}

The operator $T^\star D$ is

\begin{align}
T^\star D&=\left[\frac{1}{4}\id_{I_1O_1I_2O_2}+\frac{1}{6}\id_{I_1}\otimes X_{O_1}\otimes Z_{I_2}\otimes\id_{O_2}+\frac{1}{12}Z_{I_1}\otimes\id_{O_1}\otimes Z_{I_2}\otimes\id_{O_2}\right]\\ 
&\quad\times\left[\id_{I_1O_1}\otimes H_{I_2O_2}+H_{I_1O_1}\otimes\id_{I_2O_2}\right].
\label{eq:explicit-Tstar-D}
\end{align}

Consequently, the two operators whose characteristic polynomials are
needed are

\begin{align} 
T^\star D\Pi^{\mathrm{sym}}&=\left[\frac{1}{4}\id_{I_1O_1I_2O_2}+\frac{1}{6}\id_{I_1}\otimes X_{O_1}\otimes Z_{I_2}\otimes\id_{O_2}+\frac{1}{12}Z_{I_1}\otimes\id_{O_1}\otimes Z_{I_2}\otimes\id_{O_2}\right]\\ 
&\quad\times\left[\id_{I_1O_1}\otimes H_{I_2O_2}+H_{I_1O_1}\otimes\id_{I_2O_2}\right]\Pi^{\mathrm{sym}},\\
T^\star D\Pi^{\mathrm{asym}}&=\left[\frac{1}{4}\id_{I_1O_1I_2O_2}+\frac{1}{6}\id_{I_1}\otimes X_{O_1}\otimes Z_{I_2}\otimes\id_{O_2}+\frac{1}{12}Z_{I_1}\otimes\id_{O_1}\otimes Z_{I_2}\otimes\id_{O_2}\right]\\ 
&\quad\times\left[\id_{I_1O_1}\otimes H_{I_2O_2}+H_{I_1O_1}\otimes\id_{I_2O_2}\right]\Pi^{\mathrm{asym}}.
\label{eq:explicit-projected-Tstar-D} \end{align}

Using the Pauli relations in Eq.~\eqref{eq:explicit-gap-generators}, together with $H=A_2-A_1$, $(\Pi_A^{\mathrm{sym}})^2=\Pi_A^{\mathrm{sym}}$, $(\Pi_A^{\mathrm{asym}})^2=\Pi_A^{\mathrm{asym}}$, and $\Pi_A^{\mathrm{sym}}\Pi_A^{\mathrm{asym}}=0$, one obtains
\begin{align}
\det\left(\lambda\id-T^\star D\Pi_A^{\mathrm{sym}}\right)=\lambda^8(\lambda^2-1)^2\left(\lambda^2+\frac{1}{3}\lambda-\frac{2}{9}\right)\left(\lambda^2-\frac{1}{3}\lambda-\frac{2}{9}\right).
\label{eq:Ksym-characteristic-polynomial}
\end{align}
Similarly,
\begin{align}
\det\left(\lambda\id-T^\star D\Pi_A^{\mathrm{asym}}\right)=\lambda^{12}\left(\lambda^2+\frac{1}{3}\lambda-\frac{2}{9}\right)\left(\lambda^2-\frac{1}{3}\lambda-\frac{2}{9}\right).
\label{eq:Kasym-characteristic-polynomial}
\end{align}

Hence the nonzero eigenvalues of the symmetric sector operator are 
\begin{align}
1,1,-1,-1,\frac{2}{3},-\frac{2}{3},\frac{1}{3},-\frac{1}{3}, 
\end{align}

whereas those for the antisymmetric sector are
\begin{align}
\frac{2}{3},-\frac{2}{3},\frac{1}{3},-\frac{1}{3}. 
\end{align}

Therefore,

\begin{align} 
\norm{\sqrt{T^\star}\Pi^{\mathrm{sym}}D\Pi^{\mathrm{sym}}\sqrt{T^\star}}_1&=6,\\ \norm{\sqrt{T^\star}\Pi^{\mathrm{asym}}D\Pi^{\mathrm{asym}}\sqrt{T^\star}}_1&=2. 
\end{align}

Adding the two contributions gives

\begin{align} \norm{\sqrt{T^\star}\Pi^{\mathrm{sym}}D\Pi^{\mathrm{sym}}\sqrt{T^\star}}_1+\norm{\sqrt{T^\star}\Pi^{\mathrm{asym}}D\Pi^{\mathrm{asym}}\sqrt{T^\star}}_1=8. \end{align}

Using Eq.~\eqref{eq:explicit-gap-scal-primal}, the feasible tester $T^\star$ gives
\begin{align}
P_{\mathrm{SEQ}}\left(\mathsf C_{\mathrm{avg}}^{(2)}(t);s\right)\geq L_{\mathrm{primal}}(t)\coloneqq\frac{1}{2}+\frac{t}{4}(6+2)=\frac{1}{2}+2t.
\label{eq:Lprimal}
\end{align}

\appsubsection{Dual certificate}

We now construct a feasible point of the dual SDP for the discrimination of $C_+(t)^{\otimes2}$ and $C_-(t)^{\otimes2}$. Since $C_+(t)^{\otimes2}-C_-(t)^{\otimes2}=tD$, it is sufficient to construct the certificate for the $t$-independent operator $D$ and restore the factor $t$ at the end. The dual optimization program is 

\begin{align} 
\min_{R,Q,\lambda}\quad &\lambda\\
\text{s.t}\quad &R\succeq D,\qquad R\succeq-D,\\ 
&\tr_{O_2}R\preceq Q_{I_1O_1}\otimes\id_{I_2},\label{eq: consrtaint dual causality 1}\\ 
&\tr_{O_1}Q_{I_1O_1}\preceq\lambda\id_{I_1},\label{eq: consrtaint dual causality 2}\\ 
&Q_{I_1O_1}=Q_{I_1O_1}^\dagger,\qquad \lambda\in\mathbb{R}.\label{eq: consrtaint dual causality 3}
\end{align}

A natural starting point is the spectral majorant $|D|$, since $|D|\succeq D$ and $|D|\succeq-D$. Its partial trace contains a component proportional to $Z_{I_2}$, which does not in itself prevent dual feasibility. To construct the dual certificate below, we use the sector decomposition induced by $\Pi_{j_1}^{(1)}\otimes\Pi_{j_2}^{(2)}$ and modify only the $(+,-)$ block, retaining the absolute-value majorant in the other three sectors. The coefficients in this corrected block are chosen so that the $Z_{I_2}$ contribution cancels after tracing over $O_2$, while the two order constraints $R\succeq D$ and $R\succeq-D$ remain valid. The sector decomposition induced by $\Pi_{j_1}^{(1)}\otimes\Pi_{j_2}^{(2)}$ isolates this obstruction in the $(+,-)$ sector. We therefore retain the absolute-value majorant in the other three sectors and modify only the $(+,-)$ block. The coefficients
in this corrected block are chosen so that the unwanted $Z_{I_2}$ contribution cancels after tracing over $O_2$, while the two order constraints $R\succeq D$ and $R\succeq-D$ remain valid.

The resulting  dual ansatz is
\begin{align}
R^\star\coloneqq{}&\left[\id_{I_1O_1I_2O_2}+\left(\id_{I_1}\otimes X_{O_1}\right)\otimes\left(\id_{I_2}\otimes X_{O_2}\right)\right]\left(\Pi_+^{(1)}\otimes \Pi_+^{(2)}\right)\nonumber\\
&+\left\{\frac{7}{3}\id_{I_1O_1I_2O_2}+\left(\id_{I_1}\otimes X_{O_1}\right)\otimes\left[\id_{I_2}\otimes X_{O_2}-\frac{1}{3}\left(X_{I_2}\otimes Z_{O_2}-Y_{I_2}\otimes Y_{O_2}\right)\right]\right\}\left(\Pi_+^{(1)}\otimes \Pi_-^{(2)}\right)\nonumber\\
&+\left[\sqrt{5}\,\id_{I_1O_1I_2O_2}-\frac{1}{\sqrt{5}}H_{I_1O_1}\otimes\left(\id_{I_2}\otimes X_{O_2}\right)\right]\left(\Pi_-^{(1)}\otimes \Pi_+^{(2)}\right)\nonumber\\
&+\left[\sqrt{5}\,\id_{I_1O_1I_2O_2}+\frac{1}{\sqrt{5}}H_{I_1O_1}\otimes H_{I_2O_2}\right]\left(\Pi_-^{(1)}\otimes \Pi_-^{(2)}\right).
\end{align}

We first verify the order constraints. Both $D$ and $R^\star$ are block diagonal with respect to the orthogonal decomposition induced by $\Pi_{j_1}^{(1)}\otimes \Pi_{j_2}^{(2)}$, with $j_1,j_2\in\{+,-\}$. In the $(+,+)$, $(-,+)$, and $(-,-)$ sectors, the blocks of $R^\star$ are, respectively,
\begin{align}
R_{++}^\star=|D_{++}|,\qquad R_{-+}^\star=|D_{-+}|,\qquad R_{--}^\star=|D_{--}|.
\end{align}
For every Hermitian operator $A$, one has $|A|-A\succeq0$ and $|A|+A\succeq0$. Consequently,
\begin{align}
R_{j_1j_2}^\star\succeq D_{j_1j_2},\qquad R_{j_1j_2}^\star\succeq-D_{j_1j_2},
\end{align}
in these three sectors.

It remains to check the corrected $(+,-)$ block. On $\operatorname{ran}\Pi_+^{(1)}$, fix an eigenvalue $a\in\{+1,-1\}$ of $\id_{I_1}\otimes X_{O_1}$. On $\operatorname{ran}\Pi_-^{(2)}$, define
\begin{align}
B\coloneqq\left.(\id_{I_2}\otimes X_{O_2})\right|_{\operatorname{ran}\Pi_-^{(2)}},\qquad \hat B\coloneqq\left.\frac{1}{2}\left(X_{I_2}\otimes Z_{O_2}-Y_{I_2}\otimes Y_{O_2}\right)\right|_{\operatorname{ran}\Pi_-^{(2)}}.
\end{align}

Here, the restriction symbol means that the operators are considered only
on the indicated invariant subspace. The operator $B$ is the restriction
of the first Pauli direction, while $\hat B$ is the normalized restriction of
the second direction appearing in $H$.

The relations $A_1^2=\id$, $A_2^2=2(\id-A_3)$, and $\{A_1,A_2\}=0$ imply
\begin{align}
B^2=\hat B^2=\id,\qquad \{B,\hat B\}=0
\end{align}
on $\operatorname{ran}\Pi_-^{(2)}$. Therefore the eigenvalues of $u\id+vB+w\hat B$ are $u\pm\sqrt{v^2+w^2}$.

On the eigenspace labelled by $a$, the restrictions of the two relevant
operators are
\begin{align}
\left.D_{+-}\right|_a&=-a\id-B+2\hat B,\\
\left.R_{+-}^\star\right|_a&=\frac{7}{3}\id+aB-\frac{2a}{3}\hat B.
\end{align}
For $a=+1$, this gives
\begin{align}
\left.\left(R_{+-}^\star-D_{+-}\right)\right|_{a=+1}&=\frac{10}{3}\id+2B-\frac{8}{3}\hat B,\\
\left.\left(R_{+-}^\star+D_{+-}\right)\right|_{a=+1}&=\frac{4}{3}\id+\frac{4}{3}\hat B.
\end{align}
Their eigenvalues are, respectively,
\begin{align}
\frac{10}{3}\pm\sqrt{2^2+\left(\frac{8}{3}\right)^2}&\in\left\{0,\frac{20}{3}\right\},\\
\frac{4}{3}\pm\frac{4}{3}&\in\left\{0,\frac{8}{3}\right\}.
\end{align}
Thus both operators are positive semidefinite. For $a=-1$, one obtains
\begin{align}
\left.\left(R_{+-}^\star-D_{+-}\right)\right|_{a=-1}&=\frac{4}{3}\id-\frac{4}{3}\hat B,\\
\left.\left(R_{+-}^\star+D_{+-}\right)\right|_{a=-1}&=\frac{10}{3}\id-2B+\frac{8}{3}\hat B,
\end{align}
whose eigenvalues are again $\{0,8/3\}$ and $\{0,20/3\}$, respectively. Hence,
\begin{align}
R_{+-}^\star\succeq D_{+-},\qquad R_{+-}^\star\succeq-D_{+-}.
\end{align}
Since the four sector projectors are mutually orthogonal and sum to the identity, the sector wise inequalities imply
\begin{align}
R^\star-D\succeq0,\qquad R^\star+D\succeq0.
\end{align}
Therefore,
\begin{align}
R^\star\succeq D,\qquad R^\star\succeq-D.
\end{align}

We next verify the first dual causal constraint, Eq.~\eqref{eq: consrtaint dual causality 1}. For the $(+,+)$ block,
\begin{align}
\tr_{O_2}R_{++}^\star=\Pi_+^{(1)}\otimes\id_{I_2}+\left(\id_{I_1}\otimes X_{O_1}\right)\Pi_+^{(1)}\otimes Z_{I_2}.
\end{align}
For the corrected $(+,-)$ block,
\begin{align}
\tr_{O_2}R_{+-}^\star=\frac{7}{3}\Pi_+^{(1)}\otimes\id_{I_2}-\left(\id_{I_1}\otimes X_{O_1}\right)\Pi_+^{(1)}\otimes Z_{I_2}.
\end{align}
The terms proportional to $Z_{I_2}$ cancel, and hence
\begin{align}
\tr_{O_2}\left(R_{++}^\star+R_{+-}^\star\right)=\frac{10}{3}\Pi_+^{(1)}\otimes\id_{I_2}.
\end{align}

For the $(-,+)$ block,
\begin{align}
\tr_{O_2}R_{-+}^\star=\sqrt{5}\,\Pi_-^{(1)}\otimes\id_{I_2}-\frac{1}{\sqrt{5}}H_{I_1O_1}\Pi_-^{(1)}\otimes Z_{I_2}.
\end{align}
For the $(-,-)$ block,
\begin{align}
\tr_{O_2}R_{--}^\star=\sqrt{5}\,\Pi_-^{(1)}\otimes\id_{I_2}+\frac{1}{\sqrt{5}}H_{I_1O_1}\Pi_-^{(1)}\otimes Z_{I_2}.
\end{align}
Once again, the terms proportional to $Z_{I_2}$ cancel, giving
\begin{align}
\tr_{O_2}\left(R_{-+}^\star+R_{--}^\star\right)=2\sqrt{5}\,\Pi_-^{(1)}\otimes\id_{I_2}.
\end{align}
Adding the four sectors yields
\begin{align}
\tr_{O_2}R^\star=\left(\frac{10}{3}\Pi_+^{(1)}+2\sqrt{5}\,\Pi_-^{(1)}\right)\otimes\id_{I_2}.
\end{align}
We may therefore choose
\begin{align}
Q^\star_{I_1O_1}\coloneqq\frac{10}{3}\Pi_+^{(1)}+2\sqrt{5}\,\Pi_-^{(1)}.
\end{align}
Equivalently,
\begin{align}
Q^\star_{I_1O_1}=\left(\frac{5}{3}+\sqrt{5}\right)\id_{I_1O_1}+\left(\frac{5}{3}-\sqrt{5}\right)Z_{I_1}\otimes X_{O_1}.
\end{align}
Thus the first causal constraint, Eq.~\eqref{eq: consrtaint dual causality 1}, is satisfied with equality:
\begin{align}
\tr_{O_2}R^\star=Q^\star_{I_1O_1}\otimes\id_{I_2}\preceq Q^\star_{I_1O_1}\otimes\id_{I_2}.
\end{align}

We now check the second dual causal constraint, Eq.~\eqref{eq: consrtaint dual causality 2}. Since
\begin{align}
\tr_{O_1}\Pi_\pm^{(1)}=\id_{I_1},
\end{align}
we have
\begin{align}
\tr_{O_1}Q^\star_{I_1O_1}=\left(\frac{10}{3}+2\sqrt{5}\right)\id_{I_1}.
\end{align}
Choosing
\begin{align}
\lambda^\star\coloneqq\frac{10}{3}+2\sqrt{5},
\end{align}
the second dual causal constraint is also satisfied with equality:
\begin{align}
\tr_{O_1}Q^\star_{I_1O_1}=\lambda^\star\id_{I_1}\preceq\lambda^\star\id_{I_1}.
\end{align}
Finally, $Q^\star=(Q^{\star})^{\dagger}$ because it is a real linear combination of orthogonal projectors, and $\lambda^\star\in\mathbb{R}$. Hence $(R^\star,Q^\star,\lambda^\star)$ satisfy all dual constraints.

Recalling that $C_+^{\otimes2}-C_-^{\otimes2}=tD$ and that $t>0$, we obtain the following dual upper bound for the discrimination of the two bare channels:
\begin{align}
P_{\mathrm{SEQ}}\left(\mathsf C_{\mathrm{bare}}^{(2)}(t);s\right)\leq U_{\mathrm{dual}}(t)\coloneqq\frac{1}{2}+\frac{t}{4}\left(\frac{10}{3}+2\sqrt5\right).
\end{align}
Combining this dual upper bound with the primal lower bound in
Eq.~\eqref{eq:Lprimal}, we obtain
\begin{align}
P_{\mathrm{SEQ}}\left(\mathsf C_{\mathrm{avg}}^{(2)}(t);s\right)-P_{\mathrm{SEQ}}\left(\mathsf C_{\mathrm{bare}}^{(2)}(t);s\right)&\geq L_{\mathrm{primal}}(t)-U_{\mathrm{dual}}(t)\\
&=\frac{1}{2}+2t-\frac{1}{2}-\frac{t}{4}\left(\frac{10}{3}+2\sqrt5\right)\\
&=\frac{t}{4}\left(8-\frac{10}{3}-2\sqrt5\right)\\
&=\frac{7-3\sqrt5}{6}\,t.
\end{align}
Since $7-3\sqrt5>0$, the difference is strictly positive for every
$0<t<1/(2\sqrt5)$. Therefore,
\begin{align}
P_{\mathrm{SEQ}}\left(\mathsf C_{\mathrm{avg}}^{(2)}(t);s\right)>P_{\mathrm{SEQ}}\left(\mathsf C_{\mathrm{bare}}^{(2)}(t);s\right),
\end{align}
which proves the claimed two-use averaged-purification advantage.

\end{proof}

The analytical expression also matches the numerically optimized gap certificate $(L_{\mathrm{primal}}-U_{\mathrm{dual}})$ throughout the admissible range of $t$.

\appsection{Proof of the advantage for three or more copies}\label{app:: perfect discrimination k copies}\label{app: advantage k copies}

We now construct a pair of finite-dimensional channels whose averaged
purifications can be perfectly discriminated with three sequential uses, although the original channels cannot be perfectly discriminated with any finite number of ordinary channel uses. The strict separation therefore
persists for every finite number of bare queries.

The construction has two ingredients. First, the channels must fail the
finite-query criterion for perfect discrimination. By the characterization
of finite-query channel discrimination in Ref.~\cite{Duan09}, it is sufficient to choose channels that are not entanglement-assisted disjoint. Second, their Choi supports must be sufficiently different for the corresponding averaged purification combs to become perfectly distinguishable. The construction below uses three purified uses; it does not exclude the possibility of a different two-use perfect-discrimination
example.
Our construction requires three purified uses; it does not exclude the existence of a two-use perfect-discrimination example.

\setcounter{lemma}{1}

\begin{lemma}
    The quantum channels with Choi operators
 \begin{align}
        C_0&\coloneqq\ketbra{00}+\ketbra{11},  \\
        C_1&\coloneqq\frac12\ketbra{-0}+\ketbra{+1}+\frac12\ketbra{-1},  
    \end{align}
    are not entanglement-assisted disjoint.
\end{lemma}

\begin{proof}

We take $\supp C_0=\operatorname{span}\{\ket{00},\ket{11}\}$. The choice of $C_1$ is guided by two requirements. First, its support should intersect $\supp C_0$ non trivially, so that the two channels are not immediately perfectly distinguishable. Second, the support should still be different enough from $\supp C_0$ to allow a separation after purification.

We therefore choose $\supp C_1$ to be the hyperplane orthogonal to
$\ket{+0}$:
\begin{align}
\supp C_1\coloneqq\ket{+0}^{\perp}=\operatorname{span}\{\ket{-0},\ket{+1},\ket{-1}\}.
\end{align}
Indeed, for a vector $a\ket{00}+b\ket{11}\in\supp C_0$, the condition
$\bra{+0}(a\ket{00}+b\ket{11})=0$ gives $a=0$. Hence,
\begin{align}
\supp C_0\cap\supp C_1=\operatorname{span}\{\ket{11}\}.
\end{align}

It remains to choose a positive Choi operator supported on this
hyperplane and satisfying the trace-preserving condition. In the basis
$\{\ket{-0},\ket{+1},\ket{-1}\}$, consider
\begin{align}
C_1=a\ketbra{-0}+b\ketbra{+1}+c\ketbra{-1},
\qquad a,b,c\geq0.
\end{align}
Its partial trace over the output system is
\begin{align}
\tr_OC_1=(a+c)\ketbra{-}+b\ketbra{+}.
\end{align}
Since
\[
\id_I=\ketbra{+}+\ketbra{-},
\]
the trace-preserving condition $\tr_OC_1=\id_I$ requires
\[
b=1,\qquad a+c=1.
\]
We choose the symmetric positive solution $a=c=1/2$, which gives
\begin{align}
C_1=\frac{1}{2}\ketbra{-0}+\ketbra{+1}+\frac{1}{2}\ketbra{-1}.
\end{align}

Both Choi operators are positive, and their partial traces satisfy
\begin{align}
    \tr_{O}C_0=\id_I,\quad \tr_{O}C_1=\frac{1}{2}\ketbra{-}+\ketbra{+}+\frac{1}{2}\ketbra{-}=\id_{I}.
\end{align}

Before proving the failure of entanglement-assisted disjointness, we record several useful properties. Both channels are entanglement breaking and have commuting outputs, since they admit the measure-and-prepare representations
\begin{align}
\Ccal_0(\rho)&=\bra{0}\rho\ket{0}\ketbra{0}+\bra{1}\rho\ket{1}\ketbra{1},\\
\Ccal_1(\rho)&=\bra{+}\rho\ket{+}\ketbra{1}+\bra{-}\rho\ket{-}\frac{\id_O}{2}.
\end{align}
Thus $\Ccal_0$ is the unital completely dephasing channel in the computational basis, whereas $\Ccal_1$ measures in the $X$ basis and is nonunital. Moreover,
\begin{align}
\operatorname{rank}C_0=2,\qquad \operatorname{rank}C_1=3,\qquad \supp C_0\cap\supp C_1=\operatorname{span}\{\ket{11}\}.
\end{align}
Consequently, $d_E=3$ is the smallest common purification-environment dimension, and the common vector $\ket{11}$, equivalently the common Kraus direction $\ket{1}_O\!\bra{1}_I$, already suggests the obstruction to entanglement-assisted disjointness that we now prove.\\

Let $\Hcal_R$ be an arbitrary reference system and let
\begin{align}
\ket{\psi}_{RI}=\ket{r_0}_R\ket{0}_I+\ket{r_1}_R\ket{1}_I
\end{align}
be an arbitrary nonzero pure input. Define the linear map $R_\psi:\Hcal_I\rightarrow\Hcal_R$ by
\begin{align}
R_\psi\ket{0}_I=\ket{r_0}_R,\qquad R_\psi\ket{1}_I=\ket{r_1}_R.
\end{align}
Using the Choi representation, the corresponding output state of $\Ccal_i$ is
\begin{align}
\left(\id_R\otimes\Ccal_i\right)\left(\ketbra{\psi}\right)=\left(R_\psi\otimes\id_O\right)C_i\left(R_\psi^\dagger\otimes\id_O\right).
\end{align}
Therefore, the output support is
\begin{align}
\mathcal{S}_i(\psi)\coloneqq\supp\left[\left(\id_R\otimes\Ccal_i\right)\left(\ketbra{\psi}\right)\right]=\left(R_\psi\otimes\id_O\right)\supp C_i.
\end{align}
Using $\supp C_0=\operatorname{span}\{\ket{00},\ket{11}\}$, we obtain
\begin{align}
\mathcal S_0(\psi)=\operatorname{span}\left\{\ket{r_0}_R\ket{0}_O,\ket{r_1}_R\ket{1}_O\right\}.
\end{align}
Similarly, using
$\supp C_1=\operatorname{span}\{\ket{-0},\ket{+1},\ket{-1}\}$, we find
\begin{align}
\mathcal S_1(\psi)=\operatorname{span}\left\{(\ket{r_0}-\ket{r_1})_R\ket{0}_O,(\ket{r_0}+\ket{r_1})_R\ket{1}_O,(\ket{r_0}-\ket{r_1})_R\ket{1}_O\right\},
\end{align}
where irrelevant factors $1/\sqrt2$ have been omitted.

If $\ket{r_1}\neq0$, then
\begin{align}
\ket{r_1}_R\ket{1}_O=\frac{1}{2}\left[\left(\ket{r_0}+\ket{r_1}\right)_R\ket{1}_O-\left(\ket{r_0}-\ket{r_1}\right)_R\ket{1}_O\right]
\end{align}
belongs to both output supports. Hence,
\begin{align}
\ket{r_1}_R\ket{1}_O\in\mathcal{S}_0(\psi)\cap\mathcal{S}_1(\psi).
\end{align}
If $\ket{r_1}=0$, then $\ket{r_0}\neq0$ because $\ket{\psi}$ is nonzero, and
\begin{align}
\ket{r_0}_R\ket{0}_O\in\mathcal{S}_0(\psi)\cap\mathcal{S}_1(\psi).
\end{align}
Thus, for every reference system $\Hcal_R$ and every nonzero input $\ket{\psi}_{RI}$,
\begin{align}
\mathcal{S}_0(\psi)\cap\mathcal{S}_1(\psi)\neq\{0\}.
\end{align}
The channels are therefore not entanglement-assisted disjoint. By the
finite-query discrimination criterion of Ref.~\cite{Duan09}, they cannot
be perfectly discriminated using any finite number of bare channel uses.

That is,
\begin{align}
   P_{\mathrm{SEQ}}\left(\mathsf C_{\mathrm{bare}}^{(k)};s\right)<1.
\end{align}
\end{proof}

We now show that the averaged purification combs of these channels can
nevertheless be perfectly discriminated with three sequential uses. Applying the Weingarten expansion from Eq.~\eqref{eq: expanded permutation Scal k} gives
\begin{align}
R^{(3)}_{avg}(C_i)=\sum_{\pi,\sigma\in\mathsf S_3}\operatorname{Wg}_{d_E,3}(\pi^{-1}\sigma)\left(C_i^{\otimes3}W_\pi^A\right)\otimes W_\sigma^E.
\end{align}

For a spectral decomposition $C=\sum_a\lambda_a\ketbra{c_a}$, the operator appearing in this expression is
\begin{align}
C^{\otimes3}W_\pi^{A}=\sum_{a_1,a_2,a_3}\lambda_{a_1}\lambda_{a_2}\lambda_{a_3}\bigotimes_{\ell=1}^3\ket{c_{a_\ell}}\bra{c_{a_{\pi(\ell)}}}.
\end{align}
For $C_0=\ketbra{00}+\ketbra{11}$, this gives
\begin{align}
R_\mathrm{avg}^{(3)}(C_0)=\sum_{\pi,\sigma\in \mathsf S_3}\operatorname{Wg}_{d_E,3}(\pi^{-1}\sigma)\sum_{a_1,a_2,a_3\in\{0,1\}}\left(\bigotimes_{\ell=1}^3\ket{a_\ell a_\ell}\bra{a_{\pi(\ell)}a_{\pi(\ell)}}\right)\otimes W_\sigma^E.
\end{align}

For $C_1$, define
\begin{align}
\ket{t_0}\coloneqq\ket{-0},\qquad\ket{t_1}\coloneqq\ket{+1},\qquad\ket{t_2}\coloneqq\ket{-1},
\end{align}
so that
\begin{align}
C_1=\frac{1}{2}\ketbra{t_0}+\ketbra{t_1}+\frac{1}{2}\ketbra{t_2}.
\end{align}
Since the product of the three eigenvalues is $2^{-\#\{\ell:b_\ell\neq1\}}$, we obtain
\begin{align}
R_\mathrm{avg}^{(3)}(C_1)=\sum_{\pi,\sigma\in \mathsf S_3}\operatorname{Wg}_{d_E,3}(\pi^{-1}\sigma)\sum_{b_1,b_2,b_3\in\{0,1,2\}}2^{-\#\{\ell:b_\ell\neq1\}}\left(\bigotimes_{\ell=1}^3\ket{t_{b_\ell}}\bra{t_{b_{\pi(\ell)}}}\right)\otimes W_\sigma^E.
\end{align}

\begin{theorem}[Perfect sequential channel discrimination with purification.]
\label{app thm: perfect discrimination 3 copies}
    In a channel discrimination task, where the score function is $s(x,y)=\delta_{x,y}$, there exist channel ensembles such that, for sequential strategies,
    \begin{align}
        P_{\mathrm{SEQ}}(\mathsf{C}_\mathrm{bare}^{(k)};s)
        < P_{\mathrm{SEQ}}(\mathsf{C}_\mathrm{avg}^{(3)};s) 
        = P_{\mathrm{SEQ}}(\mathsf{C}_\mathrm{ukn}^{(3)};s) = 1.
    \end{align}
    for all finite $k$.
    \\

    In particular, let $\Hcal_I\simeq\Hcal_O\simeq\mathbb C^2$ and $\Hcal_E\simeq\mathbb C^{d_E}$, with $d_E\geq3$. Now take the hypothesis set $\mathsf{X}=\{0,1\}$ and the decision set $\mathsf{Y}=\mathsf{X}$. Encode the hypothesis set into the pair of qubit-qubit channels with Choi operators given by
    \begin{align}
        C_0&\coloneqq\ketbra{00}+\ketbra{11}, \\
        C_1&\coloneqq\frac12\ketbra{-0}+\ketbra{+1}+\frac12\ketbra{-1}, 
    \end{align}
    where $\ket{\pm}=(\ket0\pm\ket1)/\sqrt2$. 
    Then, set the channel ensembles $\mathsf{C}_\mathrm{bare}^{(k)}$, $\mathsf{C}_\mathrm{avg}^{(3)}$, and $\mathsf{C}_\mathrm{ukn}^{(3)}(\{U_x\})$
    according to the two access models in Eqs.~\eqref{eq::R1}--\eqref{eq::R3} and by taking channels $C_0$ and $C_1$ and an equal prior distribution $p(0)=p(1)=1/2$. These ensembles satisfy Eq.~\eqref{eq::inf_separation}.
\end{theorem}

\begin{proof}

We now propose the two tester outcomes
$T_0^{(3)}$ and $T_1^{(3)}$ given below. We will prove that they form
an optimal three-round sequential tester in two steps. First, we verify
that they are positive and satisfy all recursive causal-normalization
constraints. Second, we show that each outcome has zero probability on
the comb corresponding to the opposite hypothesis. The resulting
success probability is therefore equal to one. Since no discrimination
probability can exceed one, this proves that the proposed tester is
optimal. 
The relation
\[
\supp C_1=\ket{+0}^{\perp}
\]
implies that contraction with $\bra{00}$ agrees with contraction with
$-\bra{10}$ on $\supp C_1$. Applying this relation to two copies,
and using the permutation symmetry of the averaged purification comb,
produces a coherent difference of two histories that vanishes on the
support of $R_\mathrm{avg}^{(3)}(C_1)$.

The same difference does not vanish on the support of
$R_\mathrm{avg}^{(3)}(C_0)$. Indeed,
\[
\supp C_0=\operatorname{span}\{\ket{00},\ket{11}\}
\]
contains only matched input--output labels. One of the two histories is
therefore allowed by $C_0$, while the other contains an input--output
mismatch and is forbidden. The square of this coherent difference gives
a positive tester outcome identifying the first comb. Additional positive
diagonal terms are then added to complete the tester causally. These
completion terms also contain input--output mismatches and hence do not
affect the support separation.

Throughout the construction, the computational-basis vectors are ordered
as $I_3O_2I_2O_1$, and $F_{E_1E_2}$ denotes the swap operator on the
first two environment systems. Define

\begin{align}
T_0^{(3)} &=\ketbra{0}_{O_3}\otimes\id_{E_3}\otimes\Big[\ketbra{0001}_{I_3O_2I_2O_1}\otimes\id_{E_1E_2}+\ketbra{1110}_{I_3O_2I_2O_1}\otimes\id_{E_1E_2}\nonumber\\
&-\ket{0001}\!\bra{1110}_{I_3O_2I_2O_1}\otimes F_{E_1E_2}-\ket{1110}\!\bra{0001}_{I_3O_2I_2O_1}\otimes F_{E_1E_2}\Big]\otimes\ketbra{1}_{I_1}.
\label{eq:explicit-optimal-M0}
\end{align}

The complementary outcome is

\begin{align}
T_1^{(3)}&=\ketbra{1}_{O_3}\otimes\id_{E_3}\otimes\Big[\ketbra{0001}_{I_3O_2I_2O_1}\otimes\id_{E_1E_2}+\ketbra{1110}_{I_3O_2I_2O_1}\otimes\id_{E_1E_2}\nonumber\\
&-\ket{0001}\!\bra{1110}_{I_3O_2I_2O_1}\otimes F_{E_1E_2}-\ket{1110}\!\bra{0001}_{I_3O_2I_2O_1}\otimes F_{E_1E_2}\Big]\otimes\ketbra{1}_{I_1}\nonumber\\
&\quad+\id_{O_3E_3}\otimes\left(\ketbra{0101}_{I_3O_2I_2O_1}+\ketbra{1010}_{I_3O_2I_2O_1}\right)\otimes\id_{E_1E_2}\otimes\ketbra{1}_{I_1}.
\label{eq:explicit-optimal-M1}
\end{align}

The causal constraints these operators must fulfill to be feasible are given in the main text in equation \eqref{eq:SEQ_tester}. We translate them to recursive constraints for 3 copies
\begin{align}
T_0^{(3)}\succeq0,\qquad T_1^{(3)}\succeq0,\qquad T_0^{(3)}+T_1^{(3)}=\id_{O_3E_3}\otimes N^{(3)},\\
\tr_{I_3}N^{(3)}=\id_{O_2E_2}\otimes N^{(2)},\qquad\tr_{I_2}N^{(2)}=\id_{O_1E_1}\otimes\rho,\qquad\tr\rho=1.
\end{align}

\paragraph{Positivity.} The first outcome is positive because its central factor between brackets has the form $XX^\dagger$, and is
therefore positive semidefinite.  The first term of $T_1^{(3)}$ contains the same positive factor, while its second term is a sum of rank-one projectors tensored with identities and $\ketbra{1}_{I_1}$. Therefore,

\begin{align}
T_0^{(3)}\succeq0,\qquad T_1^{(3)}\succeq0.
\end{align}

\paragraph{Normalization of the outcomes.} Adding the two outcomes and using $\ketbra{0}_{O_3}+\ketbra{1}_{O_3}=\id_{O_3}$ gives

\begin{align}
T_0^{(3)}+T_1^{(3)}&=\id_{O_3E_3}\otimes\Big[\ketbra{0001}_{I_3O_2I_2O_1}\otimes\id_{E_1E_2}+\ketbra{1110}_{I_3O_2I_2O_1}\otimes\id_{E_1E_2}\nonumber\\
&\quad\quad\quad-\ket{0001}\!\bra{1110}_{I_3O_2I_2O_1}\otimes F_{E_1E_2}-\ket{1110}\!\bra{0001}_{I_3O_2I_2O_1}\otimes F_{E_1E_2}\Big]\otimes\ketbra{1}_{I_1}\\
&+\id_{O_3E_3}\otimes\left(\ketbra{0101}_{I_3O_2I_2O_1}+\ketbra{1010}_{I_3O_2I_2O_1}\right)\otimes\id_{E_1E_2}\otimes\ketbra{1}_{I_1}.
\end{align}

Thus, there exists an operator $N_\star^{(3)}$ such that
\begin{align}
T_0^{(3)}+T_1^{(3)}=\id_{O_3E_3}\otimes N_\star^{(3)}
\end{align}

\paragraph{First causal constraint.} The vectors $\ket{0001}$ and $\ket{1110}$ have orthogonal $I_3$ labels, so the cross terms vanish under the partial trace over $I_3$. Using $F_{E_1E_2}F_{E_1E_2}^{\dagger}=\id_{E_1E_2}$, we obtain

{\small
\begin{align}
\tr_{I_3}\left[\left(\ket{0001}\otimes\id_{E_1E_2}-\ket{1110}\otimes F_{E_1E_2}\right)\left(\bra{0001}\otimes\id_{E_1E_2}-\bra{1110}\otimes F_{E_1E_2}\right)\right]=\left(\ketbra{001}+\ketbra{110}\right)_{O_2I_2O_1}\otimes\id_{E_1E_2}.
\end{align}}

The diagonal completion terms give

\begin{align}
\tr_{I_3}\left[\left(\ketbra{0101}+\ketbra{1010}\right)_{I_3O_2I_2O_1}\otimes\id_{E_1E_2}\right]=\left(\ketbra{101}+\ketbra{010}\right)_{O_2I_2O_1}\otimes\id_{E_1E_2}.
\end{align}

Combining both contributions yields

\begin{align}
\tr_{I_3}N_\star^{(3)}=\left(\ketbra{001}+\ketbra{010}+\ketbra{101}+\ketbra{110}\right)_{O_2I_2O_1}\otimes\id_{E_1E_2}\otimes\ketbra{1}_{I_1}=\id_{O_2E_2}\otimes N_\star^{(2)},
\end{align}

where

\begin{align}
N_\star^{(2)}\coloneqq\left(\ketbra{01}+\ketbra{10}\right)_{I_2O_1}\otimes\id_{E_1}\otimes\ketbra{1}_{I_1}.
\end{align}

This proves the first recursive causal constraint.

\paragraph{Second causality constraint.} Tracing $N_\star^{(2)}$  over $I_2$ gives

\begin{align}
\tr_{I_2}N_\star^{(2)}=\left(\ketbra{1}_{O_1}+\ketbra{0}_{O_1}\right)\otimes\id_{E_1}\otimes\ketbra{1}_{I_1}=\id_{O_1E_1}\otimes\rho_\star,
\end{align}

where

\begin{align}
\rho_\star\coloneqq\ketbra{1}_{I_1}.
\end{align}

Finally,

\begin{align}
\rho_\star\succeq0,\qquad\tr\rho_\star=1.
\end{align}

We have therefore verified directly from $T_0^{(3)}$ and $T_1^{(3)}$ that

\begin{align}
T_0^{(3)}\succeq0,\qquad T_1^{(3)}\succeq0,\qquad T_0^{(3)}+T_1^{(3)}=\id_{O_3E_3}\otimes N_\star^{(3)},\\
\tr_{I_3}N_\star^{(3)}=\id_{O_2E_2}\otimes N_\star^{(2)},\qquad\tr_{I_2}N_\star^{(2)}=\id_{O_1E_1}\otimes\rho_\star,\qquad\tr\rho_\star=1.
\end{align}

Hence $\{T_0^{(3)},T_1^{(3)}\}$ is a valid binary three-round sequential tester.

It remains to prove that the tester has no error. Let $\Pi_1$ denote
the projector onto $\supp R_\mathrm{avg}^{(3)}(C_1)$. Since
$\supp C_1=\ket{+0}^{\perp}$, one has
\[
\bra{00}P_{\supp C_1}=-\bra{10}P_{\supp C_1}.
\]
Applying this relation to the second and third copies, and using the
invariance of $\Pi_1$ under simultaneous exchange of the first two
channel--environment copies, gives

\begin{align}
\left[\left(\bra{11}_{I_1O_1}\otimes\bra{00}_{I_2O_2}\otimes\bra{00}_{I_3O_3}\right)\otimes\id_{E_1E_2E_3}-\left(\bra{10}_{I_1O_1}\otimes\bra{11}_{I_2O_2}\otimes\bra{10}_{I_3O_3}\right)\otimes F_{E_1E_2}\otimes\id_{E_3}\right]P_1=0.
\end{align}

The expression in brackets is the contraction defining the first tester
outcome. Hence,
\begin{align}
T_0^{(3)}\Pi_1=0.
\label{eq:M0P1-zero}
\end{align}

For the other error term, let $\Pi_0$ denote the projector onto
$\supp\widetilde{\Scal}_0^{(3)}$. Since
\[
\supp C_0=\operatorname{span}\{\ket{00},\ket{11}\},
\]
every basis component in $\Pi_0$ has matching input and output labels at
each channel use. Every contribution to $T_1^{(3)}$, however, contains
an input--output mismatch on at least one copy. The coherent terms contain
either $(I_3,O_3)=(0,1)$ or $(I_1,O_1)=(1,0)$, while the two diagonal
completion terms contain respectively
$(I_2,O_2)=(0,1)$ and $(I_2,O_2)=(1,0)$. Therefore,
\begin{align}
T_1^{(3)}\Pi_0=0.
\label{eq:M1P0-zero}
\end{align}

Since $R_\mathrm{avg}^{(3)}(C_j)=\Pi_jR_\mathrm{avg}^{(3)}(C_j)\Pi_j$, Eqs.~\eqref{eq:M0P1-zero} and \eqref{eq:M1P0-zero} imply

\begin{align}
\tr\left[T_0^{(3)}R_\mathrm{avg}^{(3)}(C_1)\right]=0,\qquad \tr\left[T_1^{(3)}R_\mathrm{avg}^{(3)}(C_0)\right]=0.
\label{eq:perfect-zero-errors}
\end{align}
Consequently,
\begin{align}
    \frac{1}{2}\tr\left[T_0^{(3)}R_\mathrm{avg}^{(3)}(C_0)\right]+\frac{1}{2}\tr\left[T_1^{(3)}R_\mathrm{avg}^{(3)}(C_1)\right]=\frac{1}{2}+\frac{1}{2}=1
\end{align}

Thus, the tester discriminates the two three-round averaged purifications perfectly, and
\begin{align}
P_{\mathrm{SEQ}}\left(\mathsf C_{\mathrm{avg}}^{(3)};s\right)=1.
\end{align}
\end{proof}
\end{document}